\documentclass[nofootinbib,aps,reprint,superscriptaddress]{revtex4-2}

\usepackage[colorlinks=true,allcolors=blue]{hyperref}
\usepackage{footmisc}
\usepackage{amsthm}
\usepackage{enumitem}
\usepackage{tabularx}
\usepackage{braket}
\usepackage{amsmath}
\usepackage{amsfonts}
\usepackage{tikz-cd}
\usepackage{bbm}
\usepackage{amssymb}
\usepackage{tikz}
\usepackage{tcolorbox}

\def\basegraph{\tikz[baseline=.1ex]{
\fill (0,2ex) circle (1pt) coordinate (A);
\fill (3ex,4ex) circle (1pt) coordinate (B);
\fill (3ex,0ex) circle (1pt) coordinate (C);
\fill (6ex,2ex) circle (1pt) coordinate (D);
\fill (9ex,0) circle (1pt) coordinate (E);
\fill (9ex,4ex) circle (1pt) coordinate (F);
\fill (12ex,2ex) circle (1pt) coordinate (G);
\draw (A)--(B);
\draw (A)--(C);
\draw (D)--(B);
\draw (D)--(C);
\draw (D)--(E);
\draw (D)--(F);
\draw (G)--(E);
\draw (G)--(F);
}}

\theoremstyle{definition}
\newtheorem{definition}{Definition}
\newtheorem*{definition*}{Definition}
\newtheorem{theorem}{Theorem}
\newtheorem{lemma}{Lemma}
\newtheorem{proposition}{Proposition}
\newtheorem{corollary}{Corollary}
\newtheorem*{lemma*}{Lemma}
\newtheorem*{theorem*}{Theorem}
\newtheorem*{proposition*}{Proposition}

\theoremstyle{remark}
\newtheorem*{remark}{Remark}

\begin{document}


\title{
{
Provable quantum speedups for computing persistence in topological data analysis
}
}




\author{Casper Gyurik}
\affiliation{applied Quantum algorithms (aQa), Leiden University, 2300 RA Leiden, The Netherlands}
\author{Alexander Schmidhuber}
\affiliation{Center for Theoretical Physics, Massachusetts Institute of Technology, Cambridge, USA}
\author{Robbie King}
\affiliation{Department of Computing and Mathematical Sciences, Caltech, Pasadena, USA}
\author{Vedran Dunjko}
\affiliation{applied Quantum algorithms (aQa), Leiden University, 2300 RA Leiden, The Netherlands}
\author{Ryu Hayakawa}
\affiliation{Yukawa Institute for Theoretical Physics \& The Hakubi Center, Kyoto University, Japan}
\thanks{Contact author: ryu.hayakawa@yukawa.kyoto-u.ac.jp}

\date{\today}

\begin{abstract}
Topological data analysis (TDA) aims to extract noise-robust features from a data set by examining the number and persistence of holes in its topology. We provide an efficient quantum algorithm for 
a computational problem closely related to a core task in TDA -- determining whether a given hole persists across different length scales. Further, we prove the problem itself is $\mathsf{BQP}_1$-hard, implying that a classical solution is extremely unlikely; this stands in contrast to all previous quantum approaches to TDA, where the problems were also intractable for quantum computers, or where a rigorous proof of classical hardness still remains open. This result implies an {exponential} quantum speedup for this problem under standard complexity-theoretic assumptions. Our approach relies on encoding the persistence of a hole in a variant of the guided sparse Hamiltonian problem, where the guiding state is constructed from a harmonic representative of the hole.

\end{abstract}


\maketitle

A significant challenge in quantum computing is to find practical applications where we have strong evidence for {exponential} speedup over classical computing. 
There are only a few such promising applications of quantum computing, such as the factoring \cite{shor1999polynomial}, quantum simulation~\cite{lloyd1996universal}, and a variant of sparse-access matrix inversion~\cite{harrow2009quantum}.  
{Showing evidence for advantage can be particularly challenging.
Intuitive approaches where we exhibit a quantum algorithm with an {exponentially} faster run time than best known classical algorithms can be risky, as classical algorithms can often be improved, sometimes directly inspired by the quantum algorithm. These are so-called dequantized algorithms, which have turned many conjectured exponential speedups into polynomial speedups~\cite{tang2019quantum,chia2022sampling}
The more robust way is to show that the problem solved itself is truly hard { for classical computers in the worst-case scenario} by connecting it to a more tested plausible complexity theoretic conjecture; however, this approach is also much more challenging.}

Recently, topological data analysis (TDA) has emerged as a new application of quantum computing with a potential quantum advantage. 
Topological data analysis \cite{edelsbrunner2002topological,zomorodian:alg} has gained attention as a powerful method for addressing this challenge by utilizing tools from algebraic topology. 
{TDA is of great interest in data science, as the method is robust to noise in the data, and due to its} 
ability to capture global, higher-dimensional topological features, which traditional geometric and graph-based methods often miss~\cite{edelsbrunner2022computational}.

In topological data analysis, a set of data points is first transformed into a series of combinatorial structures called a \textit{filtration of simplicial complexes}. A \textit{simplicial complex} consists of simplices (i.e., points, lines, triangles, tetrahedra, and their higher-dimensional analogs) that are connected or ``glued'' together. 
A \textit{filtration} is a sequence of simplicial complexes, $K_1, \dots, K_n$, where each $K_{i+1}$ is formed by adding simplices to the previous complex $K_i$.
Robust information about the data is then extracted by analyzing how multi-dimensional topological features evolve throughout the filtration, {which is at the heart of modern TDA}. 
This is done using \textit{persistent homology}, a mathematical framework that tracks the emergence and disappearance of ``$k$-dimensional holes'' along the filtration. 
For example, connected components correspond to $0$-dimensional holes, loops correspond to $1$-dimensional holes, and voids correspond to $2$-dimensional holes.
The number of such holes (and the change thereof) is captured by persistent Betti numbers.
Specifically, the $k$th $(i,j)$-persistent Betti number counts the number of $k$-dimensional holes that are present in the $i$th simplicial complex and persist until the $j$th (which are typically summarized visually in a so-called barcode~\cite{ghrist:barcodes}). 
It is important to note that the more salient information is encoded in persistent holes, i.e., those that appear in several simplicial complexes along the filtration, as they are less likely to be spurious artifacts. 
As a result, calculating the \textit{persistence} of a hole -- i.e., for how many steps in the filtration it remains present -- plays a crucial role in TDA.

Classical algorithms for topological data analysis face a major challenge: computing topological features typically requires {exponential} time in the desired dimension. 
Specifically, most computational problems related to calculating $k$-dimensional features in topological data analysis have a best-known classical runtime that scales {exponentially} with $k$~\cite{zomorodian:alg, friedman:alg}.
These higher-dimensional features often capture essential higher-order interactions in data, which, as mentioned earlier, are difficult for classical algorithms to compute efficiently.
{Simultaneously, a significant challenge in quantum computing is finding practical problems where quantum algorithms can provide a provable {exponential} speedup.} The milestone result by Lloyd, Garnerone, and Zanardi (LGZ)~\cite{lloyd:tda}, which aimed to estimate normalized Betti numbers (i.e., the number of holes in a given simplicial complex divided by its size) {exponentially} faster than conventional classical algorithms, was the first result to suggest that quantum methods for topological data analysis applications could be an answer for both challenges.  
Motivated by this, there has been growing interest in topological data analysis as a promising field where quantum computers could efficiently solve practical computational problems, ideally achieving {exponential} speedups over classical methods. 
{Since the LGZ proposal, several advances have been made in quantum algorithms for the specific task of normalized Betti number estimation~\cite{berry2024analyzing, mcardle:tda, akhalwaya2023topological,nghiem2025quantum}, along with extensions to persistent Betti numbers~\cite{hayakawa:tda,ameneyro2024quantum} and other related problems \cite{schmidhuber2025quantum,zi2025quantum,leditto2023topological}. }
However, it remains unclear whether normalized (persistent) Betti number estimation will be the problem where quantum computers will have the desired useful advantages due to two key limitations:

\begin{enumerate}
    \item The practical utility of normalized Betti numbers is unclear, especially since they tend to be exponentially small~\cite{gunn2019review,mcardle:tda,schmidhuber2023complexity,berry2024analyzing}. 
    Moreover, while non-normalized persistent Betti numbers are more practically relevant, their computation is beyond the reach of quantum computers in general, as they are 
$\mathsf{QMA}_1$-hard~\cite{crichigno2022clique, king:qma} to approximate within a multiplicative error and $\#\mathsf{P}$-hard to compute exactly~\cite{schmidhuber2023complexity}.
    \item Even in the instances where quantum computers can efficiently estimate normalized (persistent) Betti numbers, there is no complexity-theoretic proof of {exponential} quantum advantage.  
    The closest known result is the $\mathsf{DQC}_1$-hardness for general simplicial complexes (i.e., not 
    naturally
    those arising in topological data analysis)~\cite{gyurik2022towards, cade2024complexity}, and it is unclear whether this extends to simplicial complexes arising in topological data analysis.
    Moreover, novel classical algorithms for normalized Betti number estimation suggest that the regime where quantum algorithms offer an {exponential} speedup is even narrower than previously believed (assuming they exist at all)~\cite{apers2023simple, berry2024analyzing}.
    We note, however, that polynomial speedups remain possible for Betti number estimation (i.e., unnormalized)~\cite{mcardle:tda}.
\end{enumerate}

\paragraph{Contributions}
In this work, we explore a new problem in topological data analysis that does not suffer from the two limitations identified above. 
{
We define and study a problem we call \textsf{Harmonic Persistence}, which is closely related to deciding whether a particular $k$-dimensional hole present in an initial simplicial complex $K_1$ persists in a larger complex $K_2$ (both given as clique complexes).
Specifically, the problem takes as input a so-called harmonic representation of a hole in $K_1$ and determines whether there are holes in $K_2$ with significantly overlapping harmonic representatives.
Our main technical result shows that for weighted simplicial complexes, the problem of \textsf{Harmonic Persistence} is $\mathsf{BQP}_1$-hard\footnote{Here, $\mathsf{BQP}_1$ denotes the perfect completeness version of $\mathsf{BQP}$ (i.e., in the YES-case the circuit accepts with certainty). 
{As the different choice of universal gate set is not known to preserve the perfect completeness \cite{rudolph2024towards}, we fix our choice of universal gate set to $\{\mathrm{CNOT}, U_{\mathrm{Pyth}}\}$ (see Eq.~\eqref{eq:pyth}) following \cite{crichigno2022clique,king:qma}.}
} and contained in $\mathsf{BQP}$. 
Those results, by assuming a standard conjecture that $\mathsf{BPP}\neq\mathsf{BQP}_1$, imply exponential quantum advantage in TDA. 
{In Table~\ref{tab:comparison}, we give a summary of previous applications of quantum TDA and evidence for speedups.
Our results can be seen as a first identification of a TDA problem where exponential quantum advantage in the worst-case is guaranteed with a plausible complexity theoretic conjecture. }

\renewcommand{\arraystretch}{1.5}
\begin{table*}
    \centering
    \begin{tabular}{|c|c|c|c|}\hline
       Task  &  Proposed Speedup & Evidence  & Reference \\ \hline \hline
       Normalized Betti number:$\beta_p/|K_p|$  & \ \ exponential \ \  & Comparison with best classical algorithm \& & \cite{lloyd:tda} \\ \cline{1-2} \cline{4-4}
       Normalized persistent Betti number: $\beta^{i,j}_p/|K^i_p|$  & exponential  & $\mathsf{DQC}_1$-hardness of generalized problem~\cite{gyurik2022towards}  & \cite{hayakawa:tda,mcardle:tda} \\ \hline
       Persistent Betti number: $\beta^{i,j}_p$ & quintic & Comparison with best classical algorithm &\cite{mcardle:tda}\\ \hline
       Harmonic Persistence  & exponential & 
       Widely-held belief that $\mathsf{BQP}_1 \nsubseteq \mathsf{BPP}$
       & This work \\ \hline
    \end{tabular}
    \caption{{Applications of quantum TDA, types of quantum speedup, and evidence for the speedup.}}
    \label{tab:comparison}
\end{table*}

Our results significantly advance our understanding of problems in topological data analysis where quantum computers can deliver practical {exponential} speedups, paving the way for genuinely useful applications of future quantum computers. 
Our proofs moreover extend the technical machinery established in~\cite{king:qma}, furthering our understanding of how to encode quantum computations in (persistent) homology. 
Finally, in our proofs we additionally define a new variant of the guided sparse Hamiltonian problem~\cite{cade:glh} that could be of independent interest (see Section~\ref{sec:final}).

}

\paragraph{Organization}
In Section~\ref{sec:prelim}, we briefly review the framework of persistent homology and introduce the technical gadgets that form the foundation of our proof techniques. 
Afterwards, in Section~\ref{sec:harmonics}, we formally define and discuss the problem of \textsf{Harmonic Persistence}, and we prove our main complexity-theoretic result. 
Finally, in Section~\ref{sec:final}, we generalize the problems introduced in Section~\ref{sec:harmonics} to define the underlying quantum primitives \textsf{Kernel Overlap} and \textsf{Low-Energy Overlap}, and provide their complexity-theoretic characterization.
In Sections~\ref{sec:discussion} and~\ref{sec:conclusion}, we give discussions and a conclusion. 

\section{Preliminaries}
\label{sec:prelim}
In this section, we outline the necessary preliminaries. 
In Section~\ref{subsec:spaces} and Section~\ref{subsec:maps}, we briefly introduce the key definitions from topological data analysis, particularly focusing on concepts related to homology.
Next, in Section~\ref{sec:gadgets}, we cover an important gadget from~\cite{king:qma} that plays a crucial role in our hardness proofs.
Finally, in Section~\ref{sec:subsetstate}, we present a method for succinctly encoding (high-dimensional) subsets of simplicial complexes, which we will use to specify the inputs for the computational problems in Section~\ref{subsec:problems}.

\subsection{Definitions of subspaces}
\label{subsec:spaces}
An abstract simplicial complex is a family of sets of vertices that is closed under taking subsets. 
In our work, we focus on so-called \textit{clique complexes} $\mathrm{cl}(G)$ associated to (weighted) graphs $G$, which are simplicial complexes whose $k$-simplices correspond to the $(k+1)$-cliques of~$G$.
Let $K$ be a simplicial complex over $n$-vertices. 
A simplex $\sigma$ is called $p$-dimensional if it is composed of $p+1$-vertices i.e., $\sigma= [v_0,v_1,...,v_p]$. 
Let $K_p\subseteq K$ be the set of $p$-dimensional simplices in $K$. The $p$-dimensional chain space $C_p$ is defined as the vector space spanned by the simplices in $K_p$. 
In this paper, we always consider vector spaces with coefficients $\mathbb{C}$. 
We identify simplices with the same vertices as being equal if their ordering of vertices is related by an even permutation and assign a negative sign if they are related by an odd permutation, i.e., 
$$
\ket{[v_0,v_1,...,v_p]}= 
\left\{
\begin{array}{ll}
\ket{[v_{\pi(0)},v_{\pi(1)},...,v_{\pi(p)}]} \\ \hspace{0.7cm}  \text{if } \pi(\cdot) \text{ is an even permutation,} \\
-\ket{[v_{\pi(0)},v_{\pi(1)},...,v_{\pi(p)}]}  \\\hspace{0.7cm}   \text{if } \pi(\cdot) \text{ is an odd permutation.}
\end{array}
\right.
$$
We suppose that the vertices of each simplex $\sigma \in K$ have some arbitrary fixed ordering, since the specifics of this ordering do not affect the topological invariants we study.
The whole chain space $C(K)$ associated with the simplicial complex $K$ is defined as the direct sum $C(K) = \bigoplus_{p = 0}^{n-1} C_p(K)$\footnote{Throughout this paper, we will use $\sigma$ to denote the actual simplices in $K$, and $\ket{\sigma}$ for the associated chains in $C(K)$.}. 
Our proof techniques necessitate the introduction of an inner product on $C(K)$, defined for any $\sigma,\tau\in K$ by 
$$
\braket{\tau|\sigma} = 
\left\{
\begin{array}{ll}
(w(\sigma))^2 & \text{ if } \sigma=\tau,\\
0 & \text{otherwise},
\end{array}
\right.
$$
where $w(\sigma)>0$ is a weight assigned for $\sigma\in K$. Following \cite{king:qma}, we use the weight function defined through the product of weight of the vertices
$$
w(\sigma) := \prod_{v\in \sigma} w(v).
$$
With this inner product, the orthonormal basis for $C_p(K)$ is given by $\{\ket{\sigma'}:= \frac{1}{w(\sigma)}\ket{\sigma}: \sigma\in K_p\}$.
The corresponding (weighted) $p$-dimensional boundary operator $\partial_p: C_p(K)\rightarrow C_{p-1}(K)$ is given by 
$$
\partial_p \ket{\sigma'}:= \sum_{i=0}^p (-1)^i w(v_i)\ket{(\sigma\backslash \{v_i\})'}
$$
for simplices $\sigma=[v_0,v_1,...,v_p] \in K_p$. 

The action of the boundary operator defines several important subspaces within the $ p $-dimensional chain space $C_p(K)$.
Firstly, the space of $ p $-dimensional cycles $ Z_p(K) $ is defined as the kernel of the boundary operator:
\[
Z_p(K) = \ker(\partial_p^K).
\]
Secondly, the $p$-dimensional boundaries $B_p(K)$ are given by the image of the $(p+1)$-dimensional boundary operator:
\[
B_p(K) = \mathrm{Im}(\partial_{p+1}^K).
\]
With these definitions in place, the $p$-dimensional homology $H_p(K)$ is then defined as the quotient space\footnote{Note that this quotient space is well-defined since $\partial^K_p \cdot \partial^K_{p+1}$ = 0.}:
\[
H_p(K) = Z_p(K) / B_p(K).
\]

The elements of $H_p(K)$ are equivalence classes of cycles, where two cycles are considered equivalent if they differ by a boundary. 
\emph{In other words, the $p$-dimensional homology captures cycles that enclose a ``hole'', ``void'' or its $p$-dimensional counterpart.}
Therefore, the dimension of the $p$-th homology group, known as the $p$-dimensional Betti number and denoted as $\beta_p$, corresponds to the number of $p$-dimensional holes in the simplicial complex. 
\emph{Moreover, the equivalence classes that span this space are represented by the cycles that enclose these holes, providing a representation of the holes within the complex.}

Finally, as we will demonstrate in the next section, $p$-dimensional holes in a simplicial complex can also be described through the $p$-dimensional harmonic homology subspace, defined as
\[
\mathcal{H}_p(K) = Z_p(K) \cap B_p(K)^{\perp}.
\]
The relationship between the $p$-th homology group and the $p$-dimensional harmonic homology subspace will be clarified in Lemma~\ref{lem:3}. 
To facilitate the spectral analysis of the harmonic homology subspace, we introduce the $p$-dimensional combinatorial Laplacian $\Delta_p: C_p(K) \rightarrow C_p(K)$, defined by
\[
\Delta_p = \partial_{p+1} \partial_{p+1}^\dagger + \partial_p^\dagger \partial_p.
\]
This operator is closely related to the $p$-dimensional harmonic homology subspace, as shown in Lemma~\ref{lemma:harm}.

\subsection{Maps between subspaces}
\label{subsec:maps}
In this section, we summarize standard results about maps between the subspaces introduced in the previous section. In the following, $K$, $K_1$ and $K_2$ are simplicial complexes such that $K_1\subseteq K_2$. 
First, we prove a standard lemma that relates the harmonic subspace to the kernel of the combinatorial Laplacian.

\begin{lemma}
\label{lemma:harm}
We have the following equality of subspaces: $\mathcal{H}_p(K) = \ker(\Delta_p^K)$.
\end{lemma}

\begin{proof}

    Because $\Delta^K_p$ is a sum of positive-semidefinite Hamiltonians $\Delta^{K,\uparrow}_p = \partial_{p+1} \partial_{p+1}^\dagger$ and $\Delta^{K,\downarrow}_p= \partial_p^\dagger \partial_p$, we know that $\Delta^K_p \ket{\sigma} = 0$ implies
    $$\Delta^{K,\uparrow}_p \ket{\sigma} =0 \text{ and } \Delta^{K,\downarrow}_p \ket{\sigma} =0.$$
    Moreover, $\Delta^{K,\downarrow}_p \ket{\sigma} =0$ implies that 
    $$
    \|\partial_p^K \ket{\sigma} \|^2 = \bra{\sigma} \Delta^{K,\downarrow}_p \ket{\sigma}=0
    $$
    and  $\Delta^{K,\uparrow}_p \ket{\sigma} =0$ implies that 
    $$
    \|(\partial_{p+1}^K)^\dagger \ket{\sigma} \|^2 = \bra{\sigma} \Delta^{K,\uparrow}_p \ket{\sigma}=0.
    $$
    
    Since  $B_p(K)^{\perp}= \ker \left((\partial_{p+1}^K)^\dagger\right)$, we thus find that $\ket{\sigma}\in \mathcal{H}_p(K)$ if $\Delta^K_p \ket{\sigma} = 0$.
    On the other hand, it clearly follows from the definition that $\ket{\sigma}\in \mathcal{H}_p(K)$ implies that $\Delta^K_p \ket{\sigma} =0$.
    
    \end{proof}
    
Next, we define a map between the homology groups of two simplicial complexes, $ K_1 \subseteq K_2 $. This map is a cornerstone in the theory of persistent homology, as it tracks how topological features evolve from one simplicial complex to another in a filtration.

\begin{lemma}
\label{lemma:ihom}
The inclusion $i_p: C_p(K_1) \rightarrow C_p(K_2)$ induces a well-defined map on homology 
\[
    i^*_p: H_p(K_1) \rightarrow H_p(K_2)
\]
given by 
\[
    z + B_p(K_1) \mapsto z + B_p(K_2).
\]
\end{lemma}

\begin{proof}

    Let $z + B_p(K_1)\in H_p(K_1)$ for $z \in Z_p(K_1)$. Then, $i^*_p(z + B_p(K_1))= z + B_p(K_2) \in H_p(K_2)$ because $z \in Z_p(K_2)$. This makes the map well-defined as claimed. 

\end{proof}

\noindent Using the map defined in Lemma~\ref{lemma:ihom}, we can now formally define what it means for a hole to ``persist''.

\begin{definition}
    An element of homology (or ``hole'' for short) $\zeta \in H_p(K_1)$  is said to \textit{persist} to the larger simplicial complex $K_2$ if $i^*_p(\zeta) \not = 0$ (i.e., its image under the inclusion map is nonzero).
\end{definition}

As our primary analysis of the holes in the simplicial complexes will focus on the harmonic homology space, it is crucial to establish a connection between this space and the standard homology. 
This link is established in the following lemma from~\cite{basu:harmonic}.

\begin{lemma}[\cite{basu:harmonic}]
\label{lem:3}

The map $\mathfrak{f}^K_p: H_p(K) \rightarrow \mathcal{H}_p(K)$ defined by 
    \[
        z + B_p(K) \mapsto \mathrm{proj}_{B_p(K)^\perp}(z), \quad z \in Z_p(K),
    \]
    is an isomorphism.
\end{lemma}

Finally, to investigate the persistence of holes through the harmonic homology subspace, we observe that the inclusion discussed in Lemma~\ref{lemma:ihom} can be suitably extended to the harmonic subspace. 
This extension is exhibited in the following lemma from~\cite{basu:harmonic}.

\begin{lemma}[\cite{basu:harmonic}]
\label{lem:4}

The restriction of $\mathrm{proj}_{B_p(K_2)^\perp}$ to $\mathcal{H}_p(K_1)$ gives a linear map 
    \[
    \mathfrak{i}_p = \mathrm{proj}_{B_p(K_2)^\perp}|_{\mathcal{H}_p(K_1)}: \mathcal{H}_p(K_1) \rightarrow \mathcal{H}_p(K_2)
    \]
    that makes the following diagram commute

    \begin{center}

    \begin{tikzcd}
    H_{p}(K_1) \arrow{r}{i^*_p} \arrow[leftrightarrow]{d}{\simeq}[swap]{\mathfrak{f}^{K_1}_p} & H_{p}(K_2) \arrow[leftrightarrow]{d}{\simeq}[swap]{\mathfrak{f}^{K_2}_p}\\
    \mathcal{H}_p(K_1) \arrow{r}{\mathfrak{i}_p} & \mathcal{H}_p(K_2)
    \end{tikzcd}
        
    \end{center}    
\end{lemma}

Importantly, Lemma~\ref{lem:3} offers a method to represent holes as elements within the harmonic homology subspace. 
We will refer to these elements, which correspond to the equivalence classes spanning the homology space, as the \textit{harmonic representative} of a hole. 

\begin{definition}[Harmonic representative of a hole]

For $\zeta \in H_p(K)$, we refer to $\mathfrak{f}^{K}_p(\zeta)$ as its \textit{harmonic representation}.

\end{definition}

Furthermore, Lemma~\ref{lem:4} provides a framework for tracking the persistence of holes by examining how the harmonic representation of a hole in $K_1$ is mapped onto the harmonic homology subspace of $K_2$.

\begin{definition}[
\cite{basu:harmonic}]
\label{def:harmpers}

A harmonic representative $\ket{\sigma} \in \mathcal{H}_p(K_1)$ is said to persist to the larger simplicial complex $K_2$ if $\mathfrak{i}_p(\ket{\sigma}) \neq 0$.

\end{definition}

\subsection{$\mathsf{QMA}_1$-hardness of homology}

In this section, we review some key aspects of the construction in~\cite{king:qma}, which establishes the $\mathsf{QMA}_1$-hardness of determining whether the homology of a simplicial complex is trivial (i.e., equal to ${0}$) or not.

The $\mathsf{QMA}_1$-hardness result is achieved via a reduction to the quantum satisfiability problem, where the objective is to determine whether a set of rank-1 projectors $\{\ket{\phi_i}\bra{\phi_i}\}^t_{i=1}$ has a non-trivial common element in their kernels. 
To facilitate this reduction, the authors of~\cite{king:qma} first construct a simplicial complex whose homology corresponds to the Hilbert space of an $n$-qubit system, $\mathbb{C}^{2^n}$.
Afterwards, the reduction introduces simplices as gadgets, where each gadget aims to ``remove" the states $\ket{\phi_i}$ from the homology. 
These gadgets effectively ensure that the remaining homology captures the solution to the original quantum satisfiability instance. 
The construction of this simplicial complex and its associated gadgets is outlined in Section~\ref{sec:gadgets}.

\subsubsection{Qubit and projector gadgets construction}
\label{sec:gadgets}

In this section, we introduce the construction of the clique complexes of \cite{king:qma} that can be used to encode instances of the quantum satisfiability problems composed of rank-1 projectors on $n$-qubits~\cite{bravyi:sat}. 
For the gadget constructions that have been discussed extensively in \cite{crichigno2022clique,king:qma}, we aim to keep the technical description brief and concise in order to allow the reader to gain a clear overview of the techniques.

Let $\{\ket{\phi_i}\bra{\phi_i}\}^t_{i=1}$ be a set of rank-1 projectors, where each $\ket{\phi_i}$ is a so-called {\it integer state}. Here, integer states are normalized states that can be written as 
$$
\ket{\phi} = \frac{1}{Z} \sum_{x\in \{0,1\}^n} a_x \ket{x}
$$
for some normalization coefficient $Z\in \mathbb{R}$ and $a_x \in \mathbb{Z}$ for all $x\in \{0,1\}^n$. 

\paragraph{Qubit graph}
We first introduce the weighted clique complex $K_1=\mathrm{cl}(\mathcal{G}_n)$ whose $(2n-1)$-th homology group $H_{2n-1}(K_1)$ is isomorphic to $\mathbb{C}^{2^n}$. 
That is, we can use the $(2n-1)$-th homology of $K_1$ to encode the space onto which an instance of the quantum satisfiability problem acts.
Specifically, $\mathcal{G}_n$ is a $7n$-vertex graph which is constructed by the graph join\footnote{The join $G_1\ast G_2$ of graphs $G_1 = (V_1, E_1)$ and $G_2 = (V_2, E_2)$ is the graph $G = (V, E)$ with vertices $V = V_1 \cup V_2$ and edges $E$ given by $E_1\cup E_2$ together with all the edges joining $V_1$ and $V_2$.} of $n$ graphs:
\begin{align} 
\mathcal{G}_n = \basegraph_1 \ast \dots \ast \basegraph_n
\end{align}
where we define the elementary graphs as
\begin{align}
    \mathcal{G} = \basegraph.
\end{align}
The Kunneth formula tells us that $H_{2n-1}(K_1)\cong \mathbb{C}^{2^n}$, see \cite{king:qma}.
The weights of all the vertices of $\mathcal{G}_n$ are~$1$. 

There is a map $s: \mathbb{C}^{2^n} \rightarrow \mathcal{H}_{2n-1}$, where $\mathbb{C}^{2^n}$ is the Hilbert space of $n$-qubit system.
Any computational basis state $\ket{x}$ where $x=x_1x_2...x_n \in \{0,1\}^n$ is mapped to 
$$\ket{\sigma_x}:=
\ket{\sigma_{x_1}}\otimes \cdots \otimes \ket{ \sigma_{x_n}}
$$
by $s$, where each of the $\ket{\sigma_{x_i}}$ is a 1-dimensional cycle in the $i$th copy of $\mathcal{G}$ (i.e., the left cycle with four edges if $x_i=0$ and right cycle if $x_i=1$). 
Also, note that $\ket{\sigma_x}$ is $(2n-1)$-dimensional (i.e., superposition of simplices with $2n$-vertices) because it is an $n$-fold tensor product of states in the $1$-dimensional chain space. 

\paragraph{Gadgets for projectors}
Since we want to construct clique complexes that encode instances of the quantum satisfiability problem, we need a way to encode projectors onto the homology of the qubit graphs.
To do so, we introduce the construction of a weighted clique complex with the input $\{\ket{\phi_i}\bra{\phi_i}\}_{i=1}^{t}$ and a weight-parameter $\lambda$. 
The simplicial complex $K_2$ that we introduce is a clique complex for a graph $\widehat{\mathcal{G}}_n$ such that $\mathcal{G}_n\subseteq \widehat{\mathcal{G}}_n$ and $K_1 \subseteq K_2$. 
Here $\widehat{\mathcal{G}}_n$ is constructed such that the chain space can be decomposed as 
$$
C_p(K_2) = C_p(K_1) \oplus C_p(\mathcal{T}_1) \oplus  \cdots \oplus C_p(\mathcal{T}_t), 
$$
where each $\mathcal{T}_i$ is a set of simplices that is introduced depending on the projector $\ket{\phi_i}\bra{\phi_i}$. 
Each of the gadget simplices $\mathcal{T}_i$ is added so that there will be a $2n$-dimensional state in $C_p(K_2)$ whose boundary is precisely $s(\ket{\phi_i})$. 
In other words, each gadget $\mathcal{T}_i$ ``fills'' the $(2n-1)$-dimensional hole corresponding to $\ket{\phi_i}$. 
While the vertices of ${\mathcal{G}}_n$ have weight $1$, the weight of the vertices in $\widehat{\mathcal{G}}_n\backslash {\mathcal{G}}_n$ is set to be $\lambda \ll 1$, such that most of the weight remains on the qubit graph. 
For more details on the precise weighting, we refer to \cite{king:qma}.

\subsection{Subset states for clique complexes}
\label{sec:subsetstate}

Chains corresponding to high-dimensional holes reside in vector spaces whose dimensions grow exponentially with the number of vertices in the simplicial complex. As a result, they generally require an exponentially large description. In this section, we define a class of high-dimensional chains that admit a succinct representation (i.e., a polynomially-sized description in the number of vertices in the simplicial complex), which we refer to as \textit{subset states}.
These subset states will serve as input to the computational problems in Section~\ref{subsec:problems}.

Let $K$ be a simplicial complex, and let $S \subseteq K_{2p-1}$ be a subset. 
We define the \textit{subset state} of $S$ as follows:

\[
\frac{1}{\sqrt{|S|}} \sum_{\sigma \in S} \ket{\sigma}.
\]

Next, consider a simplicial complex $ K = \mathrm{cl}(G) $ that admits a decomposition $ G = G_1 \ast G_2 \ast \cdots \ast G_p $. We will show that, in some cases, even when the cardinality of $ S $ is exponential, it is still possible to provide a succinct description of the corresponding subset state.
To construct the succinct description of the subset state of $ S $, consider sets of edges $ E^{(1)}, E^{(2)}, \dots, E^{(p)} $, where each $ E^{(i)} $ is a set of edges from $ G_i $. The Cartesian product $ S = E^{(1)} \times \cdots \times E^{(p)} $ forms a subset of $(2p-1)$-simplices in $ K_{2p-1} $. The subset state of this Cartesian product $ S \subseteq K_{2p-1} $ is then given by:

\[
\frac{1}{\sqrt{|S|}} \sum_{\sigma \in S} \ket{\sigma},
\]
which represents a state in $ C_{2p-1}(K) $. Equivalently, this state can be written in product form as:
\[
\ket{S} = \bigotimes_{i=1}^p \left( \frac{1}{\sqrt{|E^{(i)}|}} \sum_{e \in E^{(i)}} \ket{e} \right).
\]
Although the cardinality of $S$ can be exponential in the number of vertices, specifying the edges $\{E^{(i)}\}$ in the decomposition of $S$ provides a succinct description of $S$. 
By describing $S$ in terms of its edge sets $\{E^{(i)}\}$, we obtain a representation that is polynomial in size, even though $S$ itself may be exponentially large.

Finally, this construction can be generalized to $ m $ subsets $ S_1, S_2, \dots, S_m \subseteq K_{2p-1} $, where each $ S_i $ can be decomposed as $ S_i = E^{(1)}_i \times \dots \times E_i^{(p)} $. For simplicity, suppose that $ |S| = |S_1| = |S_2| = \cdots = |S_m| $, then we define the \textit{subset state} of $ S = \bigcup_{i=1}^m S_i $ as:

\[
\Ket{\bigcup_{i=1}^m S_i} = \frac{1}{\sqrt{m}} \sum_{i=1}^m \ket{S_i} = \frac{1}{\sqrt{m |S|}} \sum_{i=1}^m \sum_{\sigma \in S_i} \ket{\sigma}.
\]
This state is called the \textit{subset state} of $ S = \bigcup_{i=1}^m S_i $. It can be described using a total of $ m \times p $ sets of edges $ \{ E^{(j)}_i\}$, which provides a succinct (polynomial-sized) description of $S$ in terms of the number of vertices.
In the following section, we will show that the problem of deciding whether harmonics that can be represented as subset states persist or not is $\mathsf{BQP}_1$-hard and contained in $\mathsf{BQP}$.

While finding a harmonic hole is known to be $\mathsf{QMA}_1$-hard in general~\cite{crichigno2022clique, king:qma} and thus not efficiently feasible on a quantum computer, we highlight that in our formulation of $\delta\text{-}\mathsf{Harmonic}$ $\mathsf{Persistence}$, the initial hole $\sigma \in \ker \Delta^{K_1}_p$ can, in fact, be efficiently prepared using classical or quantum methods. In particular, the initial simplicial complex $K_1$ constructed in our $\mathsf{BQP}_1$-hardness proof (Proposition~\ref{prop:hard}) has a combinatorial Laplacian with a tensor product structure, enabling efficient preparation of the corresponding initial harmonic hole. Alternatively, this can be understood as the combinatorial Laplacian of $K_1$ having a sufficiently small bond dimension, allowing for efficient state preparation. Finding other methods and families of initial simplicial complexes that permit efficient hole preparation remains an interesting avenue for future research.

\section{Quantum complexity of persistent harmonics}
\label{sec:harmonics}
In this section, we present our main result. 
First, in Section~\ref{subsec:problems} we provide the formal definitions of the computational problem we study and state our main result Theorem~\ref{thm:main}, which we then prove in Section~\ref{sec:proof_main}.

\subsection{Problem definitions}
\label{subsec:problems}

Consider simplicial complexes $K_1 \subseteq K_2$.
We study the following problem: ``\textit{given a succinct classical description of a harmonic $\ket{\sigma} \in \mathcal{H}_p(K_1)$ (i.e., a representative of a hole), decide whether $||\mathrm{proj}_{\mathcal{H}_p(K_2)}(\ket{\sigma})||$ is $(a)$ large (i.e., $\Omega(1/\mathrm{poly}(n)$), or $(b)$ small (i.e., $\mathcal{O}(1/\mathrm{exp}(n)$))}''.
We formally define this problem as follows:\\

\noindent
\underline{\textbf{$\delta$-\textsf{Harmonic Persistence}}}\\[0.3em]
\noindent\begin{minipage}{\linewidth}
  \textbf{Input:} \\ 1) Simplicial complexes $K_1 \subseteq K_2$, where $K_i = \mathrm{cl}(G_i)$ and $|V(G_i)| = n$ for $i=1, 2$.\\
    2) An integer $0 \leq p < n$.\\
    3) A
   $\mathrm{poly}(n)$-size description of $S \subseteq K_{1,p}$. Here, $K_{1,p}$ is a set of $p$-simplices in~$K_1$.
   \\[1em]
   
  \textbf{Promises:} \\
   1) Spectral gap\\ $\gamma(\Delta_p^{K_2}) = \min \{|\lambda| \mid \lambda \in \mathrm{Spec}(\Delta_p^{K_2}),\text{ }\lambda > 0\}\geq \frac{1}{\mathrm{poly}(n)}$.\\
   2) For $\ket{\sigma} :=  \frac{1}{\sqrt{|S|}}\sum_{\sigma_i \in S} \ket{\sigma_i}$ either \\
  $||\mathrm{proj}_{\mathcal{H}_p(K_2)}(\ket{\sigma})|| \geq \delta$, 
  or $||\mathrm{proj}_{\mathcal{H}_p(K_2)}(\ket{\sigma})|| < \frac{1}{\mathrm{exp}(n)}$.\\[1em]
  \textbf{Output:} \\
  1 if 
   $||\mathrm{proj}_{\mathcal{H}_p(K_2)}(\ket{\sigma})|| \geq \delta$
  , or\\ 0 if $||\mathrm{proj}_{\mathcal{H}_p(K_2)}(\ket{\sigma})|| < \frac{1}{\mathrm{exp}(n)}$.
\end{minipage}

\begin{remark} $\mathrm{proj}_V$ denotes the orthogonal projection onto $V$.
\end{remark}

For harmonic representatives of holes $\ket{\sigma} \in \mathcal{H}_p(K_1)$, assessing whether $\| \mathrm{proj}_{\mathcal{H}_p(K_2)}(\ket{\sigma}) \|$ is large provides insight into whether a hole persists in $K_2$. 
Nonetheless, there are contrived cases where a harmonic hole $\ket{\sigma} \in \mathcal{H}_p(K_1)$ continues to $K_2$ in a canonical sense but only with an exponentially small overlap with $\mathcal{H}_p(K_2)$ (i.e., it gets exponentially deformed in the process).
However, such holes are arguably not the meaningful topological features in the data that TDA aims to extract.
Additionally, it is possible that a hole $\ket{\sigma} \in \mathcal{H}_p(K_1)$ does not persist in the {strict} sense as in {Definition~\ref{def:harmpers}} but instead only has support on the eigenspaces of $\Delta^{K_2}_p$ associated with exponentially small eigenvalues. 
In these situations, the quantum algorithm cannot reliably distinguish these small eigenvalues from zero, a common limitation in quantum algorithms for linear algebraic tasks. 
As a result, the algorithm may fail to detect the hole’s nonpersistence.
In short, our \textsf{Harmonic Persistence} problem may yield both false positives and false negatives compared to the canonical problem {of deciding a persistence of a hole as in Definition~\ref{def:harmpers}}. 
Although specific contrived examples can demonstrate these cases, we conjecture that these limitations will not typically occur in typical instances, such as random Vietoris-Rips complexes. 
In summary, the \textsf{Harmonic Persistence} problem provides a means to determine whether a hole meaningfully persists, specifically by determining whether there are harmonics in $K_2$ with significant overlap.
{More discussion on the definition of the problem is given in Section~\ref{sec:discussion}.}

Our main result is stated in Theorem~\ref{thm:main}.
The proof of Theorem~\ref{thm:main} is provided in Section~\ref{sec:proof_main}. 
The proof consists of two parts: first, we show the problem is $\mathsf{BQP}_1$-hard for $\delta$ as large as $1-\frac{1}{\mathrm{poly}(n)}$. Next, we show that the problem is contained in $\mathsf{BQP}$ for $\delta$ as small as $\frac{1}{\mathrm{poly}(n)}$.
Our proof utilizes and extends the framework of~\cite{king:qma} to encode quantum computation into the persistent homology of clique complexes.

\begin{theorem}
\label{thm:main}
    For all $\delta \in \left[\frac{1}{\mathrm{poly}(n)},1-\frac{1}{\mathrm{poly}(n)}\right]$, the problem $\delta\text{-}\mathsf{Harmonic}$ $\mathsf{Persistence}$ is $\mathsf{BQP}_1$-hard.
    Moreover, for all $\delta \in \left[\frac{1}{\mathrm{poly}(n)}, 1\right]$ the problem $\delta\text{-}\mathsf{Harmonic}$ $\mathsf{Persistence}$ is contained in $\mathsf{BQP}$.
\end{theorem}

As a corollary, we also obtain a complexity-theoretic characterization of deciding whether a chain has a large overlap with the harmonic homology space or not.
Intuitively, this corresponds with checking whether a chain has a large overlap with one or multiple holes.
We define this computational problem as follows. \\

\noindent\begin{minipage}{\linewidth}
{\underline{\textbf{$\delta$-\textsf{Harmonics}}}} \\[0.3em]
\textbf{Input:} \\ 1) Simplicial complex $K$ such that $K = \mathrm{cl}(G)$ and $|V(G)| = n$.\\
 2) An integer $0 \leq p < n$.\\
 3) {
A $\mathrm{poly}(n)$-size description of $S \subseteq K_{p}$. Here, $K_{p}$ is a set of $p$-simplices in~$K$.}\\[1em]

\textbf{Promises:}\\
 1) Spectral gap $\gamma(\Delta_p^{K}) = \min \{|\lambda| \mid \lambda \in \mathrm{Spec}(\Delta_p^{K}),\text{ }\lambda > 0\}\geq \frac{1}{\mathrm{poly}(n)}$.\\
 2) For
$\sigma :=  \frac{1}{\sqrt{|S|}}\sum_{\sigma_i \in S} \ket{\sigma_i}$
, either of the following holds: \\  $(a)$ $||\mathrm{proj}_{\mathcal{H}_p(K)}(\sigma)|| \geq \delta$, or $(b)$ $||\mathrm{proj}_{\mathcal{H}_p(K)}(\sigma)|| < \frac{1}{\mathrm{exp}(n)}$.\\[1em]
\end{minipage}
\noindent\begin{minipage}{\linewidth}
\textbf{Output:} \\ 1 if $||\mathrm{proj}_{\mathcal{H}_p(K)}(\sigma)|| \geq \delta$, or 0 if $||\mathrm{proj}_{\mathcal{H}_p(K)}(\sigma)|| < \frac{1}{\mathrm{exp}(n)}$.
\end{minipage}

\bigskip

As a corollary of Theorem~\ref{thm:main}, we can show the following result. 

\begin{corollary}
\label{cor:related}
    For all $\delta \in \left[\frac{1}{\mathrm{poly}(n)},1-\frac{1}{\mathrm{poly}(n)}\right]$, the problem $\delta\text{-}\mathsf{Harmonics}$ is $\mathsf{BQP}_1$-hard.
    Moreover, for all $\delta \in \left[\frac{1}{\mathrm{poly}(n)},1\right]$ the problem $\delta\text{-}\mathsf{Harmonics}$ is contained in $\mathsf{BQP}$.
\end{corollary}
\begin{proof}

$\mathsf{Harmonics}$ can be used to solve instance of $\mathsf{Harmonic}$ $\mathsf{Persistence}$ by simply setting $K = K_2$ and using the same $\ket{\sigma}$.
Analogously, the quantum algorithm for $\mathsf{Harmonic}$ $\mathsf{Persistence}$ can be used to solve $\mathsf{Harmonics}$.
    
\end{proof}

\subsection{Proof of Theorem~\ref{thm:main}}
\label{sec:proof_main}

\begin{proposition}
\label{prop:hard}
    $\delta\text{-}\mathsf{Harmonic}$ $\mathsf{Persistence}$ is $\mathsf{BQP}_1$-hard for $\delta=1-1/\mathrm{poly}(n)$.
\end{proposition}

\begin{proof}

 Let $L = \left(L_{\mathrm{yes}}, L_{\mathrm{no}}\right)$ be a promise problem in $\mathsf{BQP}_1$, let $x \in L_{\mathrm{yes}} \cup L_{\mathrm{no}}$ with $|x| = n$, and let $U_x = U_T\cdot\dots\cdot U_1$ be a quantum circuit on $m = \mathrm{poly}(n)$ qubits of size $T=\mathrm{poly}(n)$ that decides $x$.
 Specifically, we have:
 \begin{itemize}
     \item If $x \in L_{\mathrm{yes}}$, then $\mathrm{Tr}\left[\left(U_x\ket{0^m}\bra{0^m}U^\dagger_x\right) \cdot \ket{1}\bra{1}_1\right] = 1$,
     \item If $x \in L_{\mathrm{no}}$, then $\mathrm{Tr}\left[\left(U_x\ket{0^m}\bra{0^m}U^\dagger_x\right) \cdot \ket{0}\bra{0}_1\right] \geq 1-1/\mathrm{poly}(n)$.
 \end{itemize}
For reasons that will become clear in the proof of Lemma~\ref{lemma:completeness}, we will pre-idle $U_x$ by adding $L = \mathcal{O}\left(\mathrm{poly}(n)\right)$ identity gates.
That is, without loss of generality, we consider circuits of the form $U_x = U_T\cdots U_1 \cdot I^L$.

Next, we are going to utilize the construction of King \& Kohler~\cite{king:qma}.
In particular, we first use the construction of~\cite{bravyi:sat} to construct a $4$-local Hamiltonian on $\mathfrak{H} = \left(\mathbb{C}^{2}\right)^{\otimes m}_{\mathrm{comp}} \otimes \left(\mathbb{C}^{2}\right)^{\otimes 2(T+L)}_{\mathrm{clock}}$ of the form
{
\begin{equation}
\label{eq:kitaev_qsat_ham}
H_x=H_{\mathrm{in}}+ H_{\mathrm{prop}} + H_{\mathrm{clock}} + H_{\mathrm{out}}.
\end{equation}
We defer the full formal details on the construction of the Hamiltonian $H_x$ from~\cite{bravyi:sat} to Appendix~\ref{appendix:qsat}.
Importantly, $H_x$ can be written in the following form:
\begin{align}
    H_{x} &= \sum_{i=1}^{\ell} \ket{\phi_i}\bra{\phi_i},
\end{align}}
where $\ell = \mathcal{O}(\mathrm{poly}(n))$ and every $\ket{\phi_i}$ is a so-called ``integer state'' (see Sec.~\ref{sec:gadgets} and~\cite[Definition 8.1]{king:qma}).
Note that because we use the idling construction, our gate set is composed of $\{\mathrm{CNOT}, U_{\mathrm{Pyth}}, I\}$ where $U_{\mathrm{Pyth}}$ is the ``Pythagorean gate''~\cite{crichigno2022clique} 
\begin{equation}
\label{eq:pyth}
U_{\mathrm{Pyth}}=
\frac{1}{5}
\begin{pmatrix}
    3 & 4 \\
    -4 & 3
\end{pmatrix}.
\end{equation}

In the construction of the Hamiltonian of \cite{bravyi:sat} (see Appendix~\ref{appendix:qsat}), the term that is dependent on the quantum circuit is 
$$
H_{\mathrm{prop}}: = \sum_{t=1}^T H_{\mathrm{prop},t} + \sum_{t=1}^{T-1} H'_{\mathrm{prop},t},
$$
where 
\begin{align*}
   H_{\mathrm{prop},t} =\frac{1}{2} &\big[ (\ket{a1}\bra{a1}+\ket{a2}\bra{a2})_t \otimes I_{\mathrm{comp}}\\& - 
\ket{a2}\bra{a1}_t \otimes U_t - \ket{a1}\bra{a2}_t \otimes U^\dagger_t \big].  
\end{align*}
The second term $H'_{\mathrm{prop},t}$ is not dependent on the quantum circuit. 
Suppose that $t$-th gate $U_t$ is $I$. Then,
\begin{align*}
    H_{\mathrm{prop},t} &=\frac{1}{2} [ (\ket{a1}\bra{a1}+\ket{a2}\bra{a2})_t \otimes I_{\mathrm{comp}} \\&\ \ \ \  - 
\ket{a2}\bra{a1}_t \otimes I_{\mathrm{comp}} - \ket{a1}\bra{a2}_t \otimes I_{\mathrm{comp}}] \\
&= \left(\frac{1}{\sqrt{2}}(\ket{a1}-\ket{a2})\right) \left(\frac{1}{\sqrt{2}}(\bra{a1}-\bra{a2})\right) \otimes I_{\mathrm{comp}}\\
&= \left(\frac{1}{\sqrt{2}}(\ket{01}-\ket{10})\right) \left(\frac{1}{\sqrt{2}}(\bra{01}-\bra{10})\right) \otimes I_{\mathrm{comp}}
.
\end{align*}
It is clear that the construction of \cite{king:qma} can be used for this projector.

{Let us define}
\begin{align}
    \label{eq:hist}
\ket{\psi_{\mathrm{hist}}}=
\frac{1}{\sqrt{2(T+L)}} \sum_{t=1}^{T+L}\left(
\ket{C_t}\otimes \ket{\psi_{t-1}} + \ket{C'_t}\otimes \ket{\psi_t}\right),
\end{align}
where $\ket{C_t}$ and $\ket{C'_t}$ are $2(T+L)$-qubit computational basis states defined as
\begin{align*}
  & \ket{C_t}= |\underbrace{11,...,11}_{t-1},01,\underbrace{00,...,00}_{T+L-t}\rangle, \\&
\ket{C'_t}= |\underbrace{11,...,11}_{t-1},10,\underbrace{{00,...,00}}_{T+L-t}\rangle, 
\end{align*}
and $\ket{\psi_t}=U_t U_{t-1}...U_1 \ket{\psi_0}$ with $\ket{\psi_0}:= \ket{0^m}$. 
{It is known \cite{bravyi:sat} that $\ket{\psi_{\mathrm{hist}}}$ spans the 1-dimensional kernel\footnote{{In the case of $\mathsf{BQP}_1$ verifier, there is no freedom for the choice of witness register. Therefore, the history state w.r.t. the unique initial state is the only ground state of $H_{\mathrm{hist}}$.  }} of 
$$H_{\mathrm{hist}}=H_{\mathrm{in}}+ H_{\mathrm{prop}} + H_{\mathrm{clock}}.$$} 
{If $x \in L_{\text{yes}}$, $\ket{\psi_{\mathrm{hist}}}$ is not penalized by $H_{\mathrm{out}}$ at all. 
Therefore, 
we have $\ket{\psi_{\mathrm{hist}}} \in \ker H_x$ if $x \in L_{\text{yes}}$. 
On the other hand, 
if $x \in L_{\text{no}}$, 
$\ket{\psi_{\mathrm{hist}}}$ gets energy from a positive semi-definite term $H_{\mathrm{out}}$, which means that there is no frustration-free ground state of $H_x$
and therefore $\dim \ker H_x = 0$ if $x \in L_{\text{no}}$.}

We can show the following lemma regarding the minimum non-zero eigenvalue of $H_x$.
\begin{lemma}
\label{lemma:spectralgap}
    For the above $H_x$, we have $g = \gamma(H_x) \geq \frac{1}{\mathrm{poly}(n)}$ for both $x \in L_{\mathrm{yes}}$ and $x \in L_{\mathrm{no}}$.
\end{lemma}
\begin{proof}
    If $x \in L_{\mathrm{no}}$, it holds that $\gamma(H_x) \geq \frac{1}{\mathrm{poly}(n)}$ as shown in~\cite{bravyi:sat}. 
    Then, let us consider the case $x \in L_{\mathrm{yes}}$.
    First, $H_{\mathrm{hist}}=H_{\mathrm{in}}+H_{\mathrm{prop}}{+H_{\mathrm{clock}}}$ is a Hamiltonian such that $\dim\ker (H_{\mathrm{hist}})=1$ {as previously mentioned}. 
    {Moreover, $g'= \gamma(H_{\mathrm{hist}})\geq \frac{1}{\mathrm{poly}(n)}$, see Lemma~\ref{lemma:gap_hist} in Appendix~\ref{appendix:qsat}}.  
    The full Hamiltonian is
    $H_x= H_{\mathrm{hist}} + H_{\mathrm{out}}$, where $H_{\mathrm{out}}$ is also a projector. 
    {Let $\ket{\psi_{\mathrm{hist}}}$ be the ground state of $H_{\mathrm{hist}}$ as defined in Eq.~\eqref{eq:hist}.} 
    For $x \in L_{\mathrm{yes}}$, we know that $\dim\ker (H_x)=1$ and $\ket{\psi_{\mathrm{hist}}}\in \ker (H_x)$. 
    Therefore, for any normalized state $\ket{\varphi}$ in $(\ker H_x)^\perp$, it holds that 
    $$
    \bra{\varphi} H_x \ket{\varphi} \geq \bra{\varphi} H_{\mathrm{hist}} \ket{\varphi} \geq g',
    $$
    which concludes the proof of the lemma. 
\end{proof}

Let $N=m+ 2(T+L)$ be the total number of qubits. 
Using the notation on pg.\ 9 of~\cite{king:qma}, we construct two simplicial complexes $K_1 = \mathrm{cl}(\mathcal{G}_N)$ and $K_2 = \mathrm{cl}(\widehat{\mathcal{G}}_N)$ on $\mathrm{poly}(n)$ vertices by using vertex weight $\lambda = c \ell^{-1} g$, where using $\ast$ to denote the graph join we have
\begin{align} 
\label{eq:base_graph}
\mathcal{G}_N = \basegraph_1 \ast \dots \ast \basegraph_N 
\end{align}
and $\widehat{\mathcal{G}}_N$ is obtained through adding the gadgets to $\mathcal{G}_N$
 for every term $\ket{\phi_i}\bra{\phi_i}$ in $H_x$ (see~\cite[Section 10]{king:qma} for more details). 
Specifically, for these simplicial complexes it holds that $K_1 \subset K_2$ and that there is an isomorphism of Hilbert spaces $\mathrm{s}: \mathfrak{H} \rightarrow C_{2N-1}(K_1)$ whose image we denote $s(\mathfrak{H}) = \mathfrak{H}_{\mathrm{sim}}$ and which we refer to as the ``simulated qubit space''.

\begin{lemma}
\label{lemma:harmonic}

Consider 
$$\ket{\psi_{\mathrm{prehist}}}=
\frac{1}{\sqrt{2L}} \left(\sum_{t=1}^{L}
\ket{C_t}\otimes \ket{0^m} + \ket{C'_t}\otimes \ket{0^m}\right).$$
Then, $\ket{\sigma_{\mathrm{prehist}}}$ is a subset state with an efficient description. 
Moreover, $\ket{\sigma_{\mathrm{prehist}}} \in \mathcal{H}_{2N-1}(K_1)$.
\end{lemma}

\begin{proof}
    Any computational basis state $\ket{x} \in \mathfrak{H}$ where $x=x_1 x_1 ... x_{N} \in \{0,1\}^N$ is mapped by $s$ to 
    $$s(\ket{x})=\ket{\sigma_x}:=
    \ket{\sigma_{x_1}}\otimes \cdots \otimes \ket{ \sigma_{x_N}}
    $$
    where $\ket{\sigma_{x_i}}$ represents a 1-dimensional cycle corresponding to the qubit gadget graph, i.e., a left cycle for $\ket{0}$ and a right cycle for $\ket{1}$ in the $i$th copy of the base graph
    $$\basegraph .$$
    Because $\sigma_{x_i}$ is a uniform superposition of 4-number of 1-simplices, $\ket{\sigma_x}$ is a uniform superposition of $4^N$-number of $(2N-1)$-simplices. 
    Using the notation in Sec.~\ref{sec:subsetstate}, $s(\ket{x})$ is a subset state for the subset
    $$S_x:=E^{(1)}_{x_1}\times E^{(2)}_{x_2}\times \cdots \times E^{(N)}_{x_N},$$
    where $E^{(i)}_{x_i}$ are sets of four edges corresponding to the cycles. 
    Then, 
    $$\ket{\sigma_{\mathrm{prehist}}}= \frac{1}{\sqrt{2L}}
    \left(\sum_{t=1}^{L} 
s(\ket{C_t,0^m}) + s(\ket{C'_t,0^m})\right)
  $$ 
    is 
    a uniform superposition of $2L\cdot4^N$-number of $(2N-1)$-dimensional simplices in $K_1$. 
    $\ket{\sigma_{\mathrm{prehist}}}$ can be efficiently described by 
    $$
    S_{C_1,0^m}, S_{C_1',0^m},...,S_{C_L,0^m}, S_{C_L',0^m}.
    $$
    The proof of $\ket{\sigma_{\mathrm{prehist}}} \in \mathcal{H}_{2N-1}(K_1)$ is straightforward because $\partial_{2N-1}^{K_1} \ket{\sigma_{\mathrm{prehist}}}=0$ holds by construction and $(\partial_{2N}^{K_1})^\dagger\ket{\sigma_{\mathrm{prehist}}}=0$ also holds because there are no $2N$-dimensional simplices in $K_1$.
\end{proof}

\begin{lemma}
\label{lemma:spectral_gap}
$\gamma(\Delta_{2N-1}^{K_2}) = \Omega\left(\lambda^{18} \ell^{-1} g\right) \geq \frac{1}{\mathrm{poly}(n)}$.
\end{lemma}
\begin{proof}
The case $\dim \mathcal{H}_{2N-1}(K_2) = 0$ follows from the promise gap established in~\cite[Theorem 10.7]{king:qma}.
Next, let us consider the case $\dim \mathcal{H}_{2N-1}(K_2) \geq 1$. 
Suppose $|\sigma\rangle$ is a state orthogonal to $\mathcal{H}_{2N-1}(K_2)$ yet with energy $\langle\sigma|\Delta^{K-2}_{2N-1}|\sigma\rangle \leq \gamma(\Delta_{2N-1}^{K_2})$.
From Lemma \ref{lemma:overlap} it follows that $|\sigma\rangle$ has at least $0.99$ overlap with $s(\ker H_x)$. 
Next, we note that $s(\ker H_x)$ is a $\mathcal{O}(g)$-perturbation of $\mathcal{H}_{2N-1}(K_2)$ (see~\cite[Definition 7.15]{king:qma}). 
Finally, using Part 1 of \cite[Lemma 7.18]{king:qma} we find that the projector onto $s(\ker H_x)$ is within $\mathcal{O}(g)$ of the projector onto $\mathcal{H}_{2N-1}(K_2)$ in operator norm. 
Thus, $|\sigma\rangle$ also overlaps with $\mathcal{H}_{2N-1}(K_2)$, leading to a contradiction.
\end{proof}

\bigskip

\begin{lemma}
    \label{lemma:completeness}
    If $x \in L_{\mathrm{yes}}$, then there exists sufficiently large $L\in \mathrm{poly}(n)$ such that $||\mathrm{proj}_{\mathcal{H}_{2n-1}^{K_2}}(\ket{\sigma_{\mathrm{prehist}}})||^2 \geq 1-\frac{1}{\mathrm{poly}(n)}$.
\end{lemma}
\begin{proof}
Since $x \in L_{\mathrm{yes}}$, we know $\ket{\psi_{\mathrm{hist}}}$ in Eq.~\eqref{eq:hist} is in $\ker H_x$.
Following the arguments of~\cite{cade:glh}, we note that 
\begin{align}
    \label{eq:overlap}
    |\langle \psi_{\mathrm{prehist}}| \psi_{\mathrm{hist}} \rangle|^2 = \frac{L}{L + T } = 1 - \frac{T}{L+T},
\end{align}
which for sufficiently large $L=\mathcal{O}(\mathrm{poly}(n))$ can be made to scale as $1 - \frac{1}{r(n)}$, for any choice of polynomial~$r(n)$.

Let $\ket{\sigma_{\mathrm{hist}}} = s(\ket{\psi_{\mathrm{hist}}})$ and $\ket{\sigma_{\mathrm{prehist}}} = s(\ket{\psi_{\mathrm{prehist}}})$.
Since $s$ is an isomorphism of Hilbert spaces, we get from Eq.~\eqref{eq:overlap} that
\begin{align}
    \label{eq:simulated_overlap}
    |\langle \sigma_{\mathrm{prehist}}| \sigma_{\mathrm{hist}} \rangle|^2 \geq 1 - \frac{1}{\mathrm{poly}(n)}
\end{align}
To finish, we note that since $s(\ker H_x)$ is an $\mathcal{O}(g)$-perturbation of $\mathcal{H}_{2n-1}(K_2)$ (see Lemma~\ref{lemma:overlap}) we know that we can decompose
\begin{align}
    \label{eq:perturb_decomp}
    \ket{\sigma_{\mathrm{hist}}} = \ket{\sigma_{\mathrm{harm}}} + \ket{\sigma_{\mathrm{rest}}},
\end{align}
for some $\ket{\sigma_{\mathrm{harm}}} \in \mathcal{H}_{2n-1}(K_2)$ and $||\ket{\sigma_{\mathrm{rest}}}||= \mathcal{O}(g)$.
Using the above decomposition, we find that
\begin{align}
||&\mathrm{proj}_{\mathcal{H}_{2n-1}^{K_2}}(\ket{\sigma_{\mathrm{prehist}}})||^2 \\ &\geq |\langle \sigma_{\mathrm{prehist}}| \sigma_{\mathrm{harm}}\rangle|^2\\
    &= |\langle \sigma_{\mathrm{prehist}}| \big(|\sigma_{\mathrm{harm}}\rangle + \ket{\sigma_{\mathrm{rest}}} - \ket{\sigma_{\mathrm{rest}}}\big)|^2\\
    &= |\langle \sigma_{\mathrm{prehist}}| \sigma_{\mathrm{hist}} \rangle - \langle \sigma_{\mathrm{prehist}}|\sigma_{\mathrm{rest}}\rangle|^2\\
    &\geq ||\langle \sigma_{\mathrm{prehist}}| \sigma_{\mathrm{hist}} \rangle| - |\langle \sigma_{\mathrm{prehist}}|\sigma_{\mathrm{rest}}\rangle| |^2\\
    &\geq 1 - \frac{1}{\mathrm{poly}(n)}.
\end{align}

\end{proof}

\bigskip

\begin{lemma}
    \label{lemma:soundness}
    If $x \in L_{\mathrm{no}}$, then $||\mathrm{proj}_{\mathcal{H}_{2n-1}^{K_2}}(\ket{\sigma})|| = 0$.
\end{lemma}
\begin{proof}[Proof of Lemma~\ref{lemma:soundness}]

If $x \in L_{\mathrm{no}}$, then $\dim \ker H_x = 0$ and by Theorem 1.3 of~\cite{king:qma} we have $\mathcal{H}_{2n-1}^{K_2} = \{0\}$.
This implies that $\mathrm{proj}_{\mathcal{H}_{2n-1}^{K_2}}(.)$ is the map that maps everything to zero and thus that $||\mathrm{proj}_{\mathcal{H}_{2n-1}^{K_2}}(\ket{\sigma})|| = 0$.

\end{proof}

In conclusion, by combining all of the above Lemmas we find that the instance of $\mathsf{Harmonic}$ $\mathsf{Persistence}$ with input $(K_1, K_2)$, $2n-1$ and $\ket{\sigma_{\mathrm{prehist}}}$ allows us to decide whether $x \in L_{\mathrm{yes}}$ or $x \in L_{\mathrm{no}}$.

\end{proof}

\bigskip

\begin{proposition}
\label{prop:algorithm}
    For all $\delta \in \left[\frac{1}{\mathrm{poly}(n)}, 1\right]$ the problem $\delta\text{-}\mathsf{Harmonic}$ $\mathsf{Persistence}$ is in $\mathsf{BQP}$.
\end{proposition}

\begin{proof}

We apply quantum pase estimation on $U = e^{i\Delta_p^{K_2}}$ with $\ket{\sigma}$ prepared in the eigenvector register, and we measure the eigenvalue register up to precision $\frac{1}{2}\gamma(\Delta^{K_2}_p)$, where $\gamma(\Delta^{K_2}_p)$ is defined in the promise statement of \textsf{Harmonic Persistence} (for more details on the implementation see~\cite[Thm 11.2]{king:qma}). 
We repeat this a total of $\mathcal{O}\left(\mathrm{poly}(n)\right)$ many times, and we output 1 if we encounter a zero eigenvalue at any point, and otherwise we output 0 if we encounter no zero eigenvalues at all. Since $\gamma(\Delta^{K_2}_p) \geq \frac{1}{\mathrm{poly}(n)}$, and $\delta\geq1/\mathrm{poly}(n)$, this quantum algorithm correctly distinguishes $||\mathrm{proj}_{\mathcal{H}_p(K)}(\ket{\sigma})|| \geq \delta$ from $||\mathrm{proj}_{\mathcal{H}_p(K)}(\ket{\sigma})|| < \frac{1}{\mathrm{exp}(n)}$ using a quantum circuit of depth $\mathcal{O}(\mathrm{poly}(n))$.

\end{proof}

\section{Underlying quantum primitives}
\label{sec:final}

In this section, we generalize the problem introduced in Section~\ref{sec:harmonics} by removing the restriction of the input to be only combinatorial Laplacians.
In Section~\ref{subsec:problems_underlying}, we formally define the computational problems arising from this generalization and present our results. 
The proofs of these results are given in Section~\ref{subsec:proof_underlying1} and~\ref{subsec:proof_underlying2}.

\subsection{Problem definitions}
\label{subsec:problems_underlying}

The primary approach of this paper is to conceptualize the problem of \textsf{Harmonic Persistence} as an instance of the guided sparse Hamiltonian problem, wherein the Hamiltonian is represented by the combinatorial Laplacian, and the guiding state is the harmonic hole whose persistence we seek to analyze. 
However, this analysis is not entirely straightforward as the underlying quantum primitive is different from the traditional local Hamiltonian problem upon which the guided sparse Hamiltonian problem relies. 
Instead, the underlying primitives of \textsf{Harmonic Persistence} are what we call \textsf{Kernel overlap} or \textsf{Low-energy overlap} outlined below.

\medskip

\noindent\begin{minipage}{\linewidth}
{\underline{\textbf{\textsf{Kernel} \textsf{overlap}}}} \\[0.3em]
\textbf{Input:} \\ 1) An $n$-qubit log-local\footnotemark[1] Hamiltonian $H = \sum_{i=1}^m H_i$, with $m = \mathcal{O}(\mathrm{poly}(n))$.\\
 2) A succint description of a quantum state $\ket{\psi}$ (e.g., a ``subspace state'').\\[1em]

\textbf{Promises:}\\
 1) Spectral gap $\gamma(H) = \min \{|\lambda| \mid \lambda \in \mathrm{Spec}(H),\text{ }\lambda > 0\}\geq \frac{1}{\mathrm{poly}(n)}$.\\
\ 2) Either of the following holds:\\ $(a)$ $||\mathrm{proj}_{\ker(H)}(\ket{\psi})|| \geq \frac{1}{\mathrm{poly(n)}}$, or\\ $(b)$ $||\mathrm{proj}_{\ker(H)}(\ket{\psi})|| < \frac{1}{\mathrm{exp}(n)}$.\\[1em]
\end{minipage}
\noindent\begin{minipage}{\linewidth}
\textbf{Output:} \\ 1 if $||\mathrm{proj}_{\ker(H)}(\ket{\psi})|| \geq \frac{1}{\mathrm{poly(n)}}$, or\\ 0 if $||\mathrm{proj}_{\ker(H)}(\ket{\psi})|| < \frac{1}{\mathrm{exp}(n)}$.
\end{minipage}

\medskip

\noindent\begin{minipage}{\linewidth}
{\underline{\textbf{\textsf{Low-energy} \textsf{overlap}}}} \\[0.3em]
\textbf{Input:} \\1) An $n$-qubit log-local\footnotemark[1] Hamiltonian $H = \sum_{i=1}^m H_i$, with $m = \mathcal{O}(\mathrm{poly}(n))$.\\
 2) A threshold $\eta \in \Omega(1/\mathrm{poly}(n))$.\\
 3) A succinct description of a quantum state $\ket{\psi}$ 
.\\[1em]

\textbf{Promise:}\\
 Either of the following holds: \\  $(a)$ $||\mathrm{proj}_{E_{\leq \eta}(H)}(\ket{\psi})|| \geq \frac{1}{\mathrm{poly(n)}}$, or \\$(b)$ $||\mathrm{proj}_{E_{\leq \eta}(H)}(\ket{\psi})|| < \frac{1}{\mathrm{exp}(n)}$.
\\[1em]
\end{minipage}
\noindent\begin{minipage}{\linewidth}
\textbf{Output:} \\ 1 if $||\mathrm{proj}_{E_{\leq \eta}(H)}(\ket{\psi})|| \geq \frac{1}{\mathrm{poly(n)}}$, or\\ 0 if $||\mathrm{proj}_{E_{\leq \eta}(H)}(\ket{\psi})|| < \frac{1}{\mathrm{exp}(n)}$.
\end{minipage}
\begin{remark}
    Here we define $E_{\leq \eta}(H) = \mathrm{span}_{\mathbb{C}}\big\{\ket{\psi}\text{ } | \text{ }\ket{\psi}\text{ eigenvector with eigenvalue $< \eta$}\big\}$.
\end{remark}

As previously noted, these generalized versions of the \textsf{Harmonic Persistence} problem differ from the guided sparse Hamiltonian problem examined in studies such as \cite{cade:glh}. In the guided sparse Hamiltonian problem, a guiding state is provided that is guaranteed to have a significant overlap with the ground state in both YES and NO instances. 
In contrast, the \textsf{Kernel Overlap} and \textsf{Low-energy Overlap} problems involve a succinct description of a quantum state, and the task is to determine whether this state has overlap with the zero- or low-energy subspace. 
This fundamental distinction separates the \textsf{Kernel Overlap} and \textsf{Low-energy Overlap} problems from the guided sparse Hamiltonian problem.
Interestingly, one way to relate the original guided sparse Hamiltonian problem to \textsf{Low-energy Overlap} is through the application of a quantum singular value transform~\cite{gilyen:qsvt} to the input Hamiltonian $H$ to turn it into the projector onto the $E_{\leq \eta}(H)$.

The overlap problem is particularly noteworthy because, unlike the guided sparse Hamiltonian problem, it remains hard even when both the state $\ket{\psi}$ and the Hamiltonian $H$ are sparse. 
Specifically, while sparsity enables efficient classical evaluation of $\bra{\psi}H\ket{\psi}$, it does not necessarily facilitate the estimation of the overlap of $\ket{\psi}$ with the low-energy subspace.
Finally, we establish the complexity characterizations for the \textsf{Kernel Overlap} and \textsf{Low-energy Overlap} problems, the proofs of which are provided in Sections~\ref{subsec:proof_underlying1} and \ref{subsec:proof_underlying2}.

\footnotetext[1]{Theorems~\ref{thm:underlying1} and~\ref{thm:underlying2} still hold for Hamiltonians with a constant locality (the proofs are just slightly more convoluted).}

\begin{theorem}
\label{thm:underlying1}
    $\mathsf{Kernel}$ $\mathsf{overlap}$ is $\mathsf{BQP}_1$-hard and contained in $\mathsf{BQP}$.
\end{theorem}

\begin{theorem}
\label{thm:underlying2}
    $\mathsf{Low}$-$\mathsf{energy}$ $\mathsf{overlap}$ is $\mathsf{BQP}$-complete.
\end{theorem}

\subsection{Proof of Theorem~\ref{thm:underlying1}}
\label{subsec:proof_underlying1}

\begin{proposition}
\label{prop:underlying1_hard}
    \textsf{Kernel Overlap} is $\mathsf{BQP}_1$-hard.
\end{proposition}

\begin{proof}
 Let $L = \left(L_{\mathrm{yes}}, L_{\mathrm{no}}\right)$ be a promise problem in $\mathsf{BQP}_1$, let $x \in L_{\mathrm{yes}} \cup L_{\mathrm{no}}$ with $|x| = n$, and let $U_x = U_T\cdot\dots\cdot U_1$ be a quantum circuit on $m = \mathrm{poly}(n)$ qubits of depth $T=\mathrm{poly}(n)$ that decides $x$.
 Specifically, we have:
 \begin{itemize}
     \item If $x \in L_{\mathrm{yes}}$, then $$\mathrm{Tr}\left[\left(U_x\ket{0^m}\bra{0^m}U^\dagger_x\right) \cdot \ket{1}\bra{1}_1\right] = 1,$$
     \item If $x \in L_{\mathrm{no}}$, then $$\mathrm{Tr}\left[\left(U_x\ket{0^m}\bra{0^m}U^\dagger_x\right) \cdot \ket{0}\bra{0}_1\right] \geq 2/3$$
 \end{itemize}
We pre-idle $U_x$ by adding $L = \mathcal{O}\left(\mathrm{poly}(n)\right)$ identity gates.
That is, without loss of generality we consider circuits of the form $U_x = U_T\cdots U_1 \cdot I^L$.
Next, we consider a Hamiltonian on $\mathfrak{H} = \mathbb{C}^{2^n} \otimes \mathbb{C}^{T+L}$ defined as
\begin{align}
\label{eq:kitaev_ham}
    H_x = H_{\mathrm{in}} + H_{\mathrm{out}} + H_{\mathrm{prop}}, 
\end{align}
where
\begin{align}
    H_{\mathrm{in}} &= \left(\sum_{i=1}^m \ket{1}_i\bra{1}_i \right) \otimes \ket{0}\bra{0},\\
    H_{\mathrm{out}} &= \ket{0}\bra{0} \otimes \ket{T+L}\bra{T+L},\\
    H_{\mathrm{prop}} &= \sum_{t=1}^{T+L} I \otimes \ket{t-1}\bra{t-1} + I \otimes \ket{t}\bra{t} \\ &- U_t\otimes \ket{t}\bra{t-1} - U_t^\dagger \otimes \ket{t-1}\bra{t}.
\end{align}
{Note that $H_x$ in this section is different from that in Section~\ref{sec:harmonics}.}
We assume that the clock is represented by binary strings with $O(\log(T+L))$ qubits. 
Finally, we define the ``history state'' given by
\begin{align}
    \label{eq:history_underlying}
    \ket{\psi_{\mathrm{hist}}} = \frac{1}{\sqrt{T+L+1}}\sum_{t = 0}^{T+L} U_t\cdots U_1\ket{0^m} \otimes \ket{t}.
\end{align}
and note that $\gamma(H_x) = \min \{|\lambda| \mid \lambda \in \mathrm{Spec}(H),\text{ }\lambda > 0\}\geq \frac{1}{\mathrm{poly}(n)}$, and that the following properties hold:
\begin{itemize}
    \item $\ket{\psi_{\mathrm{hist}}} \in \ker H_x$ if $x \in L_{\mathrm{yes}}$,
    \item $\ker H_x = \{0\}$ if $x \in L_{\mathrm{no}}$.
\end{itemize}
{The fact that $\ker H_x = \{0\}$ if $x \in L_{\mathrm{no}}$ implies that there is no state that have overlap with $\ker H_x$ in this case.} 
To finalize the proof, we use the arguments of~\cite{cade:glh} and note that for 
\begin{align}
    \label{eq:prehist_underlying}
    \ket{\psi} = \frac{1}{\sqrt{L+1}}\sum_{t = 0}^{L} \ket{0^m} \otimes \ket{t}.
\end{align}
it holds that 
\begin{align}
    \label{eq:overlap_underlying}
    |\langle \psi| \psi_{\mathrm{hist}} \rangle|^2 = \frac{L+1}{L + T + 1} = 1 - \frac{T}{L+T+1},
\end{align}
which for sufficiently large $L=\mathcal{O}(\mathrm{poly}(n))$ can be made to scale as $1 - \frac{1}{r(n)}$, for any choice of polynomial~$r(n)$.  
In conclusion, the instance of $\mathsf{Kernel}$ corresponding to $H_x$ and $\ket{\psi}$ allows us to decide $x \in L_{\mathrm{yes}}$ or $x \in L_{\mathrm{no}}$.

\end{proof}

\begin{proposition}
\label{prop:underlying1_algorithm}
    \textsf{Kernel Overlap} is in $\mathsf{BQP}$.
\end{proposition}

\begin{proof}

Analogous to Proposition~\ref{prop:algorithm}, we apply QPE on $U = e^{iH}$ with $\ket{\psi}$ prepared in the eigenvector register, and we measure the eigenvalue register up to precision $\gamma(H)$.
We repeat this $\mathcal{O}\left(\mathrm{poly}(n)\right)$ many times, and we output 1 if we encounter a zero eigenvalue, and we output 0 if we have not encountered a zero eigenvalue.

\end{proof}

\subsection{Proof of Theorem~\ref{thm:underlying2}}
\label{subsec:proof_underlying2}
\begin{proposition}
\label{prop:underlying2_hard}
    $\mathsf{Low}$-$\mathsf{energy}$ $\mathsf{overlap}$ is $\mathsf{BQP}$-hard.
\end{proposition}

\begin{proof}

Let $L = \left(L_{\mathrm{yes}}, L_{\mathrm{no}}\right)$ be a promise problem in $\mathsf{BQP}$, let $x \in L_{\mathrm{yes}} \cup L_{\mathrm{no}}$ with $|x| = n$, and let $U_x = U_T\cdot\dots\cdot U_1$ be a quantum circuit on $m = \mathrm{poly}(n)$ qubits of depth $T=\mathrm{poly}(n)$ that decides $x$ {after a standard error reduction for $\mathsf{BQP}$ by parallel repetition}.
 Specifically, we have:
 \begin{itemize}
     \item If $x \in L_{\mathrm{yes}}$, then $$\mathrm{Tr}\left[\left(U_x\ket{0^m}\bra{0^m}U^\dagger_x\right) \cdot \ket{1}\bra{1}_1\right] {\geq 1-\mathcal{O}(2^{-n})},$$
     \item If $x \in L_{\mathrm{no}}$, then $$\mathrm{Tr}\left[\left(U_x\ket{0^m}\bra{0^m}U^\dagger_x\right) \cdot \ket{0}\bra{0}_1\right] {\geq 1- \mathcal{O}(2^{-n})}$$
 \end{itemize}
We pre-idle $U_x$ by adding $L = \mathcal{O}\left(\mathrm{poly}(n)\right)$ identity gates.
That is, without loss of generality we consider circuits of the form $U_x = U_T\cdots U_1 \cdot I^L$.
Next, consider the Hamiltonian on $\mathfrak{H} = \mathbb{C}^{2^n} \otimes \mathbb{C}^{T+L}$ from Eq.~\eqref{eq:kitaev_ham}.

Finally, we again consider {$H_x = H_{\mathrm{in}} + H_{\mathrm{out}} + H_{\mathrm{prop}}$ as in Eq.~\eqref{eq:kitaev_ham} and} the ``history state'' $\ket{\psi_{\mathrm{hist}}}$ from Eq.~\eqref{eq:history_underlying} and $\ket{\psi}$ from Eq.~\eqref{eq:prehist_underlying} and note that they satisfy Eq.~\eqref{eq:overlap_underlying}.
{
We note that
\begin{align*}
&\mathrm{Tr}[ H_{\mathrm{out}}\ket{\psi_{\mathrm{hist}}}\bra{\psi_{\mathrm{hist}}}]\\&=
\frac{1}{T+L+1} \mathrm{Tr}\left[\left(U_x\ket{0^m}\bra{0^m}U^\dagger_x\right) \cdot \ket{0}\bra{0}_1\right] \\&
\begin{cases}
    \leq \frac{1}{T+L+1}\cdot\mathcal{O}(2^{-n}) & \text{if } x \in L_{\mathrm{yes}}, \\
    \geq \frac{1}{T+L+1}\cdot(1-\mathcal{O}(2^{-n})) & \text{if } x \in L_{\mathrm{no}}.
\end{cases}
\end{align*}
Therefore, we can choose $\eta=\min \{\frac{1}{2(T+L+1)}, \gamma(H_{\mathrm{hist}})\}\geq 1/\mathrm{poly}(n)$ where $\gamma(H_{\mathrm{hist}})$ is the minimun non-zero eigenvalue of $H_{\mathrm{hist}}$, which is known to be larger than $1/\mathrm{poly}(n)$ \cite{gharibian2012hardness}, so that
}
\begin{itemize}
    \item $\ket{\psi_{\mathrm{hist}}} \in E_{\leq \eta}(H_x)$ if $x \in L_{\mathrm{yes}}$,
    \item $E_{\leq \eta}(H_x) = \{0\}$ if $x \in L_{\mathrm{no}}$.
\end{itemize}
In conclusion, $\mathsf{Low}$-$\mathsf{energy}$ $\mathsf{overlap}$ with input $(H_x, \eta, \ket{\psi})$ allows us to decide whether $x \in L_{\mathrm{yes}}$ or $x \in L_{\mathrm{no}}$.

\end{proof}

\begin{proposition}
\label{prop:underlying2_algorithm}
    $\mathsf{Low}$-$\mathsf{energy}$ $\mathsf{overlap}$ is in $\mathsf{BQP}$.
\end{proposition}

\begin{proof}

Analogous to Proposition~\ref{prop:algorithm}, we apply quantum phase estimation on $U = e^{iH}$ with $\ket{\psi}$ prepared in the eigenvector register, and we measure the eigenvalue register up to precision $\eta$.
We repeat this $\mathcal{O}\left(\mathrm{poly}(n)\right)$ many times, and we output 1 if we encounter a zero eigenvalue, and we output 0 if we have not encountered a zero eigenvalue.

\end{proof}

\section{Discussion}
\label{sec:discussion}

In this section, we give several discussions and open questions. 

\subsection{On the definition of the harmonic persistence problem}

{Our problem of \textsf{Harmonic Persistence}
naturally characterize the problem of deciding the persistence of a hole given as a harmonic representative. 
However, as it is novelly characterized as {\it overlap estimation problems} of an input vector across subspaces of interest, 
 \textsf{Harmonic Persistence} differs from the {traditional} problem of determining the persistence of a hole in several subtle ways.}
Firstly, if a harmonic hole in $K_1$ persists to $K_2$ but gets exponentially deformed in the process, the original harmonic in $K_1$ might have an exponentially small overlap with the harmonic representative of $K_2$. 
However, such holes are arguably not the meaningful topological features in the data that TDA aims to extract.   
Secondly, it may happen that holes do not fully persist but instead move slightly out of the harmonic space of $K_2$, acquiring an exponentially small eigenvalue rather than zero. 
In these cases, the quantum algorithm may not distinguish this small eigenvalue from zero (a limitation present in many quantum algorithms for linear algebraic problems), meaning it mistakenly thinks the hole persists. 

In short, the \textsf{Harmonic Persistence} problem we study may produce both false positives and false negatives in contrived cases relative to the canonical problem. 
Although contrived examples can demonstrate such cases, we conjecture that these issues will not arise in typical instances (such as random Vietoris-Rips complexes).

Finally, while we conjecture that a similar result could hold for unweighted simplicial complexes, our current proof techniques rely on weighting. Nevertheless, weighted simplicial complexes have recently garnered significant attention in topological data analysis applications~\cite{kara:weighted, wang:weighted, sharma:weighted, baccini:weighted}, suggesting that the weighted variant is also valuable in practice.

\subsection{Path to practical quantum advantage in TDA}
{
While our result establishes an exponential quantum speedup in TDA based on a plausible complexity theoretic conjecture in the worst-case setting, the next important step is to strengthen this situation to a more practical setting for data that appear in real-world applications.  In this section, we discuss steps towards more practical quantum advantage in TDA. 
}
    \paragraph{Computing the topological summary}

{In many pipelines in TDA, persistent Betti numbers are computed for the construction of topological representation or summary of the data called barcode, persistence diagrams, persistence landscapes, or Betti curves. 
One interesting application of the \textsf{Harmonic Persistence} problem is the computation of such topological summaries because \textsf{Harmonic Persistence} can be used to estimate when a specific bar ends. 
\textsf{Harmonic Persistence} can be efficiently applied for the computation of topological summary when the number of initial holes is polynomially bounded, and we can somehow specify all the initial holes.
}

{
However, there is a dichotomy between quantum advantage and initial hole identification for this application. 
For example, in the early stage of the filtration, the underlying graph of the Vietoris-Rips complex is supposed to be sparse, and therefore holes are easier to find. 
It is certainly possible to apply our quantum algorithm to such sparse complexes. 
However, our results do not imply classical hardness for clique complexes for such sparse graphs. 
There is even a possibility that there is an efficient classical algorithm to check the persistence in the filtration for such sparse input. 
Our work raises the important question of if it is possible to show quantum advantage for the computation of topological summary or exclude the possibility of quantum advantage by showing an efficient classical algorithm.
}

\paragraph{Tracking a deformation of a hole across a filtration}
{Additionally, it is also important to find direct applications of \textsf{Harmonic Persistence} itself. 
A potential approach in this direction is to extract information about the shape evolution or deformation of a specific hole in the filtration. 
It would be interesting to extract information about the deformation of the hole in terms of essential contents~\cite{basu:harmonic} and the spatial location of the hole
\cite{salch2021mathematics}.
}

   \paragraph{Identification of initial holes}

{
In this paper, we have shown that the \textsf{Harmonic Persistence} problem is $\mathsf{BQP}_1$-hard when we are given a succinct description of the harmonic representative of an initial hole. 
However, it is unclear whether we have access to the description of the hole to check the persistence in practical situations. 
We illustrate two open directions about the identification of initial holes.

One direction is to consider instances where it is indeed easy to find an initial hole from just the description of the given simplicial complex. 
An example of such an instance would be a clique complex of a sparse graph. It should be easy to find holes when the number of simplices is small. 
However, our hardness reduction fails for such sparse instances. 
At the same time, we do not know an efficient classical algorithm for such instances as well. It is an open question whether one can show quantum advantage for the \textsf{Harmonic Persistence} problem starting from sparse simplicial complexes. 

The second approach would be to consider random instances for clique complexes with (exponentially) many {\it known} high-dimensional holes. 
For example, one can imagine simplicial complexes introduced in \cite{berry2024analyzing} using complete graphs or gadget clique complexes in \cite{crichigno2022clique,king:qma,hayakawa2025computational}. 
Then, there is a way of picking a random initial hole from the known holes according to certain distributions. 
Now, we can argue the persistence of the picked hole after randomly adding simplices according to certain distributions. 
This approach is complementary to the previous approach because now we are not trying to find initial holes, but trying to randomly choose one hole from the many known holes.
}

\subsection{Related works}

While our problems and results share some similarities with the guided Hamiltonian problem~\cite{cade:glh}, they are fundamentally different. 
Although our problem includes a succinct description of a quantum state as part of the input, this should not be mistaken for a guiding state in the local Hamiltonian problem. 
Specifically, the state we provide is not an initial guess for the ground state of a Hamiltonian but rather the state corresponding to a hole whose persistence we seek to evaluate. 
{Our problem involves estimating
the ground state overlap of this state with respect to a different Hamiltonian,}
rather than using it to determine whether that Hamiltonian has a small ground state energy.

\subsection{Future work}
Our work opens several interesting directions for future work. 
Firstly, an obvious direction is to establish $\mathsf{BQP}_1$-hardness and containment in $\mathsf{BQP}$ for the \textsf{Harmonic Persistence} problem in the context of \textit{unweighted} clique complexes. 
While the construction in~\cite{crichigno2022clique} utilizes unweighted simplicial complexes, it is unclear how this approach could yield unweighted complexes with a sufficiently large spectral gap of the combinatorial Laplacian -- an essential property for containment in $\mathsf{BQP}$.
Secondly, in Theorem~\ref{thm:underlying2}, we establish the $\mathsf{BQP}$-completeness of the \textsf{Low-Energy Overlap} problem. 
This problem generalizes \textsf{Harmonic Persistence}, in part by lifting the restriction that the input must be a combinatorial Laplacian. 
However, it remains an open question whether $\mathsf{BQP}$-completeness can be proved for the \textsf{Low-Energy Overlap} problem when restricted specifically to combinatorial Laplacians of clique complexes.
Finally, recent developments in persistent homology have introduced the concept of the \textit{persistent Laplacian}~\cite{wang2020persistent, memoli:persistent}. 
The dimension of the kernels of these persistent Laplacians is known to correspond to the persistent Betti numbers. 
This connection raises the following natural question: can we define a problem related to the persistent Laplacian that is $\mathsf{BQP}_1$-hard and contained in $\mathsf{BQP}$, akin to the \textsf{Harmonic Persistence} problem? 

Beyond these complexity-theoretic directions, an important avenue for future work is to develop concrete applications of \textsf{Harmonic Persistence} to real-world TDA pipelines. 
Key challenges include identifying practically relevant families of simplicial complexes where initial holes can be efficiently specified, and demonstrating a quantum speedup for the computation of topological summaries such as barcodes and persistence diagrams beyond the worst-case setting established in this work.

\section{Conclusion}
\label{sec:conclusion}

In this paper, we explored the computational complexity of a computational problem closely related to the core task in TDA of determining whether a given hole persists across different length scales.
We defined our computational problem as determining whether the projection of a harmonic representative of a hole from one simplicial complex $K_1$ to a larger simplicial complex $K_2$ is large or small.
Our main result establishes that the formalized version of this problem, which we refer to as $\delta\text{-}\mathsf{Harmonic}$ $\mathsf{Persistence}$, is $\mathsf{BQP}_1$-hard and contained in $\mathsf{BQP}$.
Additionally, we introduced a related computational problem, $\delta\text{-}\mathsf{Harmonics}$, which assesses whether a chain has significant overlap with the harmonic homology space, and showed that this problem is also $\mathsf{BQP}_1$-hard and contained in $\mathsf{BQP}$.
Finally, we generalized our approach by introducing $\textsf{Kernel overlap}$ and $\textsf{Low-energy overlap}$ as new underlying quantum primitives, which differ subtly from the guided sparse Hamiltonian problem while maintaining a similar computational complexity. 
Our findings highlight the connections between TDA and quantum computing, providing a foundation for further research in this area.

\begin{acknowledgments}
CG, VD and RH thank Nicolas Sale for early discussions on the topic. AS thanks Seth Lloyd for helpful discussions. AS was partially funded by NSF grant PHY-2325080. RK was supported by the IQIM, an NSF Physics Frontiers
Center (NSF Grant PHY-1125565). RH was supported by JSPS KAKENHI PRESTO Grant Number JPMJPR23F9 and MEXT KAKENHI Grant Number 21H05183, Japan.
This work was also partially supported by the Dutch Research Council (NWO/OCW), as part of the Quantum Software Consortium programme (project number 024.003.03), and co-funded by the European Union (ERC CoG, BeMAIQuantum, 101124342). 
Views and opinions expressed are however those of the author(s) only and do not necessarily reflect those of the European Union or the European Research Council. 
Neither the European Union nor the granting authority can be held responsible for them. 
This work was also supported by the Dutch National Growth Fund (NGF), as part of the Quantum Delta NL programme.
\end{acknowledgments}

\appendix

\section{Technical details on construction of Hamiltonian from~\cite{bravyi:sat}}
\label{appendix:qsat}



Consider a quantum circuit of the form $U = U_T\cdots U_1$. We summarize a construction of projectors on $4$-local Hamiltonian on $\mathfrak{H} = \left(\mathbb{C}^{2}\right)^{\otimes n}_{\mathrm{comp}} \otimes \left(\mathbb{C}^{4}\right)^{\otimes T}_{\mathrm{clock}}=\mathfrak{H}_{\mathrm{comp}}\otimes \mathfrak{H}_{\mathrm{clock}}$  that appears with the construction of \cite{bravyi:sat}. 
Each of the 4-dimensional clock qudit is spanned by 
$$
\ket{a1}, \ket{a2}, \ket{d}, \ket{u},
$$
where these labels represent {\it active phase 1, active phase 2, dead, unborn}. 
Define a ``legal clock states'' for $t=1,2,...,T$ as 
$$
\ket{C_t}:= |\underbrace{d,...,d}_{t-1},a1,\underbrace{u,...,u}_{T-t}\rangle, \ \ 
\ket{C'_t}:= |\underbrace{d,...,d}_{t-1},a2,\underbrace{u,...,u}_{T-t}\rangle.
$$
We also introduce the history state as
\begin{align*}
&\ket{\psi_{\mathrm{hist}}{(w)}}\\& =
\frac{1}{\sqrt{2T}} \left(\sum_{t=1}^T
\ket{C_t}\otimes \ket{\psi_{t-1}{(w)}} + \ket{C'_t}\otimes \ket{\psi_t{(w)}}\right), 
\end{align*}
where $\ket{\psi_{0}{(w)}}:= \ket{w}\otimes \ket{0...0}$ for a witness $\ket{w}$. 

The Hamiltonian constructed in \cite{bravyi:sat} is in the form of 

$$
H=H_{\mathrm{in}}+ H_{\mathrm{prop}} + H_{\mathrm{clock}} + H_{\mathrm{out}}
$$
Each of the terms are defined as follows:
first, 
$$
H_{\mathrm{in}}:= \ket{a1}\bra{a1}_1 \otimes \left( \sum_{b\in R_{\mathrm{in}}} \ket{1}\bra{1}_b \right), 
$$
where $R_{\mathrm{in}}$ represents the qubits in the computation register other than the witness register. 
Next, 
$$
H_{\mathrm{prop}}: = \sum_{t=1}^T H_{\mathrm{prop},t} + \sum_{t=1}^{T-1} H'_{\mathrm{prop},t},
$$
where 
\begin{align*}
H_{\mathrm{prop},t} =\frac{1}{2} &\big[ (\ket{a1}\bra{a1}+\ket{a2}\bra{a2})_t \otimes I_{\mathrm{comp}} \\& - 
\ket{a2}\bra{a1}_t \otimes U_t - \ket{a1}\bra{a2}_t \otimes U^\dagger_t \big] 
\end{align*}
and
\begin{align*}
    H'_{\mathrm{prop},t} =\frac{1}{2}& (\ket{a2,u}\bra{a2,u}+\ket{d,a1}\bra{d,a1} \\& -\ket{d,a1}\bra{a2,u}-\ket{a2,u}\bra{d,a1})_{t,t+1} \otimes I_{\mathrm{comp}}.
\end{align*}
Third, 
$$H_{\mathrm{clock}}:= \sum_{j=1}^6 H^{(j)}_{\mathrm{clock}},$$
where
\begin{align*}
    &H^{(1)}_{\mathrm{clock}} := \ket{u}\bra{u}_1,\\
    &H^{(2)}_{\mathrm{clock}} := \ket{d}\bra{d}_T\\
    &H^{(3)}_{\mathrm{clock}} \\ &  \sum_{1\leq i<k\leq T} (\ket{a1}\bra{a1}+\ket{a2}\bra{a2})_i \otimes (\ket{a1}\bra{a1}+\ket{a2}\bra{a2})_k \\
    &H^{(3)}_{\mathrm{clock}}:= \\ &  \sum_{1\leq i<k\leq T} (\ket{a1}\bra{a1}+\ket{a2}\bra{a2})_i \otimes (\ket{a1}\bra{a1}+\ket{a2}\bra{a2})_k \\
    &H^{(4)}_{\mathrm{clock}} := \sum_{1\leq i<k\leq T}  (\ket{a1}\bra{a1}+\ket{a2}\bra{a2} + \ket{u}\bra{u})_i\otimes \ket{d}\bra{d}_k\\
    &H^{(5)}_{\mathrm{clock}} := \sum_{1\leq i<k\leq T}  \ket{u}\bra{u}_j \otimes (\ket{a1}\bra{a1}+\ket{a2}\bra{a2} + \ket{d}\bra{d})_k \\
    &H^{(6)}_{\mathrm{clock}} :=\sum_{1\leq i\leq T-1} \ket{d}\bra{d}_i \otimes \ket{u}\bra{u}_{i+1}
\end{align*}

Finally, 
$$H_{\mathrm{out}}:= \ket{a2}\bra{a2}_T \otimes \ket{1}\bra{1}_{\mathrm{out}}. $$

Then, the following results are shown in \cite{bravyi:sat}.
$\ket{\psi_{\mathrm{hist}}}$ is in the kernel of $H(U)=H_{\mathrm{in}}+ H_{\mathrm{prop}} + H_{\mathrm{clock}}$ for any witness state. 
Moreover, if $U$ accepts a witness with probability $1$, $\ker H$ is non-empty (for YES instances). 
If $U$ accepts any witness at most with probability $\epsilon=1/\mathrm{poly}(n)$, the ground state energy is larger than $1/\mathrm{poly}(n)$ (for NO instances). 

Note that we can represent these states using two qubits instead of the 4-dimensional qudit. 
We assign $\ket{00}$ to $\ket{u}$, $\ket{01}$ to $\ket{a1}$, $\ket{11}$ to $\ket{d}$. In this clock qubit representation, proper clock states can be represented as 
\begin{align*}
 & \ket{C_t}:= |\underbrace{11,...,11}_{t-1},01,\underbrace{00,...,00}_{T-t}\rangle, \\  &
\ket{C'_t}:= |\underbrace{11,...,11}_{t-1},10,\underbrace{00,...,00}_{T-t}\rangle.  
\end{align*}


{
Finally, we prove the following lemma.
\begin{lemma}
\label{lemma:gap_hist}
Let $
H_{\mathrm{hist}}=H_{\mathrm{in}}+ H_{\mathrm{prop}} + H_{\mathrm{clock}}
$. Then, it holds that 
$$\gamma(H_{\mathrm{hist}})\geq \frac{1}{\mathrm{poly}(n)}$$
where $\gamma(H_{\mathrm{hist}})$ is the minimum non-zero eigenvalue of $H_{\mathrm{hist}}$.
\end{lemma}
}
\begin{proof}
{
We can only consider legal clock space because $\gamma(H_{\mathrm{hist}})$ keeps legal/illegal clock space invariant and illegal clock space is penalized with a constant energy by $H_{\mathrm{clock}}$. 
Let $\Pi$ be a projector onto legal clock space tensored with $\mathfrak{H}_{\mathrm{comp}}$. 
We will consider the spectrum of $\Pi H_{\mathrm{in}}\Pi+\Pi H_{\mathrm{prop},t} \Pi+\Pi H'_{\mathrm{prop},t}\Pi$ restricted to the legal clock space. 
Then, 
\begin{align*}
\Pi H_{\mathrm{prop},t} \Pi=\frac{1}{2} &\big[ (\ket{C_t}\bra{C_t}+\ket{C_t'}\bra{C_t'}) \otimes I_{\mathrm{comp}} \\& - 
\ket{C_t'}\bra{C_t} \otimes U_t - \ket{C
_t}\bra{C_t'} \otimes U^\dagger_t \big],
\end{align*}
\begin{align*}
    \Pi H'_{\mathrm{prop},t}\Pi =\frac{1}{2}& (\ket{C_t'}\bra{C_t'}+\ket{C_t}\bra{C_t} \\& -\ket{C_t}\bra{C_t'}-\ket{C_t'}\bra{C_T}) \otimes I_{\mathrm{comp}},
\end{align*}
and 
$$
\Pi H_{\mathrm{in}}\Pi:= \ket{C_1}\bra{C_1} \otimes \left( \sum_{b\in R_{\mathrm{in}}} \ket{1}\bra{1}_b \right). 
$$
Now, consider a modification of quantum circuit as 
$$U' = U_T\cdot I\cdots I\cdot U_2\cdot I\cdot U_1,$$
which can be obtained by inserting ``idling'' gates between every gates in the original circuit. 
Then, it can be seen that 
$$\sum_{t=1}^T (\Pi H_{\mathrm{prop},t} \Pi+\Pi H'_{\mathrm{prop},t} \Pi)$$ forms a standard propagation term as appeared in Eq.~\eqref{eq:kitaev_ham} applied for a modified circuit $U'$.
($\Pi H_{\mathrm{prop},t}\Pi$ can be seen as the propagation term for the original gates and $\Pi H'_{\mathrm{prop},t}\Pi$ can be seen as the propagation term for the inserted idling gates.)
}
{Therefore, we can apply the same argument with Lemma 23 in \cite{gharibian2012hardness} to conclude the lemma.}
\end{proof}

\section{Overlap lemma}
\label{appendix:overlap_lemma}

\begin{lemma}
\label{lemma:overlap}
Let $H$ be an $m$-local Hamiltonian on $n$ qubits of the form
\[
H = \sum_{i=1}^t \phi_i
\]
where each term $\phi_i$ is a rank-1 projector onto an integer state $\ket{\phi_i}$. 
Also, suppose that $H$ has a non-empty kernel, and let $g$ denote the spectral gap of $H$. 
Next, let $\widehat{\mathcal{G}}_n$ denote the weighted graph described in \cite[Section 10.3]{king:qma} obtained by adding the gadgets for every term $\phi_i$ to the base graph $\mathcal{G}_n$ from Eq.~\eqref{eq:base_graph} using vertex-weighting parameter 
\[\
\lambda = c t^{-1} g
\]
for some sufficiently small constant $c > 0$ Finally, let $K = \mathrm{cl}(\widehat{\mathcal{G}}_n)$ and $E  =  c \lambda^{4m+2} t^{-1} g$.

Then, the following holds: for any $\ket{\sigma} = s(\ket{\varphi}) \in \mathfrak{H}_{\mathrm{sim}}$ such that $\bra{\sigma}\Delta^K_{2n-1}\ket{\sigma} \leq E$ we have that $\ket{\varphi}$ has overlap at least $0.99$ with $\ker H$.
\end{lemma}

\begin{proof}
\begin{widetext}
We import the notation described in \cite[Definition 10.8]{king:qma}. 
Using this notation, we recall that
\begin{align} 
\Pi_0^k + \sum_{i=1}^t \Pi_i^k &= \text{id} \label{complete_projectors_eq} \\
\Pi^{(\mathcal{A})}_i + \Pi^{(\mathcal{B})}_i + \hat{\Phi}_i + \hat{\Phi}^\perp_i &= \Pi^{2n-1}_0 + \Pi^{2n-1}_i \label{sing_gadg_eigenspaces_eq}
\end{align}
Using the above equations, we derive the following equalities.
\begin{align}
1 &= \langle\sigma|\sigma\rangle \\
&= \langle\sigma|\big(\Pi_0^{2n-1} + \sum_{i=1 }^t\Pi_i^{2n-1}\big)|\sigma\rangle \\
&= \sum_{i=1 }^t \langle\sigma|\big(\Pi_0^{2n-1} + \Pi_i^{2n-1}\big)|\sigma\rangle - (t-1) \langle\sigma|\Pi_0^{2n-1}|\sigma\rangle \\
&= \sum_{i=1 }^t \langle\sigma|\big(\hat{\Phi}^\perp_i + \hat{\Phi}_i + \Pi^{(\mathcal{A})}_i + \Pi^{(\mathcal{B})}_i\big)|\sigma\rangle - (t-1) \langle\sigma|\Pi_0^{2n-1}|\sigma\rangle \\
&= \Bigg( \langle\sigma| \sum_{i=1 }^t \hat{\Phi}^\perp_i|\sigma\rangle - (t-1) \langle\sigma|\Pi_0^{2n-1}|\sigma\rangle \Bigg) + \langle\sigma| \sum_{i=1 }^t \hat{\Phi}_i|\sigma\rangle + \langle\sigma| \sum_{i=1 }^t \Pi^{(\mathcal{A})}_i|\sigma\rangle + \langle\sigma| \sum_{i=1 }^t \Pi^{(\mathcal{B})}_i|\sigma\rangle
\end{align}
Examining the first term in the above equation, we notice that
\begin{align}
&\langle\sigma| \sum_{i=1 }^t \hat{\Phi}^\perp_i |\sigma\rangle - (t-1) \langle\sigma|\Pi_0^{2n-1}|\sigma\rangle \\
&= \langle\sigma| \sum_{i=1}^t \big(\Phi^\perp_i + \mathcal{O}(\lambda)\big) |\sigma\rangle - (t-1) \langle\sigma|\Pi_0^{2n-1}|\sigma\rangle \\
&= \langle\sigma| \sum_{i=1}^t \Phi^\perp_i |\sigma\rangle - (t-1) \langle\sigma|\Pi_0^{2n-1}|\sigma\rangle + \mathcal{O}(\lambda t) \\
&= \langle\sigma| \sum_{i=1}^t \big(\Pi^{(\mathcal{H})}_n - \Phi_i\big) |\sigma\rangle - (t-1) \langle\sigma|\Pi_0^{2n-1}|\sigma\rangle + \mathcal{O}(\lambda t) \\
&= \langle\sigma|\Pi^{(\mathcal{H})}_n|\sigma\rangle - \langle\sigma| \sum_{i=1}^t \Phi_i |\sigma\rangle - (t-1) \langle\sigma|\big(\Pi_0^{2n-1} - \Pi^{(\mathcal{H})}_n\big)|\sigma\rangle + \mathcal{O}(\lambda t)
\end{align}
where we used that $\hat{\Phi}^\perp_i = \Phi^\perp_i + \mathcal{O}(\lambda)$, which follows from~\cite[Lemma 50]{king:qma} and~\cite[Lemma 33, Part 1]{king:qma}

Next, we note that $\sum_{i=1}^t \Phi_i$ is precisely the implementation of the Hamiltonian $H$ on the simulated $n$-qubit subspace $\mathcal{H}_n$, so the term $\langle\varphi| \sum_{i=1}^t \Phi_i |\varphi\rangle$ is equal to
\begin{equation}
\langle\sigma| \sum_{i=1}^t \Phi_i |\sigma\rangle = \langle\varphi|H|\varphi\rangle
\end{equation}
Putting this all together, we have
\begin{align}
1 &= \langle\sigma|\Pi^{(\mathcal{H})}_n|\sigma\rangle - \langle\varphi|H|\varphi\rangle - (t-1) \langle\sigma|\big(\Pi_0^{2n-1} - \Pi^{(\mathcal{H})}_n\big)|\sigma\rangle + \mathcal{O}(\lambda t) \\
& \qquad + \langle\sigma| \sum_{i=1}^t \hat{\Phi}_i|\sigma\rangle + \langle\sigma| \sum_{i=1}^t \Pi^{(\mathcal{A})}_i|\sigma\rangle + \langle\sigma| \sum_{i=1}^t \Pi^{(\mathcal{B})}_i|\sigma\rangle
\end{align}
Finally, we use~\cite[Lemmas 10.9, 10.10, 10.11]{king:qma} to bound the final three terms. 
Doing so, we get
\begin{equation}
1 = \langle\sigma|\Pi^{(\mathcal{H})}_n|\sigma\rangle - \langle\varphi|H|\varphi\rangle - (t-1) \langle\sigma|\big(\Pi_0^{2n-1} - \Pi^{(\mathcal{H})}_n\big)|\sigma\rangle + g/100
\end{equation}
\begin{equation}
\implies \ \langle\sigma|\Big(\mathbbm{1} - \Pi^{(\mathcal{H})}_n\Big)|\sigma\rangle 
+ \langle\varphi|H|\varphi\rangle \leq g/100
\end{equation}
\begin{align}
\implies \ \langle\sigma|\Big(\mathbbm{1} - \Pi^{(\mathcal{H})}_n\Big)|\sigma\rangle &\leq g/100 \\
\langle\varphi|H|\varphi\rangle &\leq g/100
\end{align}
\end{widetext}
for sufficiently small constant $c > 0$. 
However, the Hamiltonian $H$ has a spectral gap $g$ by assumption, which shows that $\ket{\phi}$ has overlap at least $0.99$ with $\ker H$. 
\end{proof}

\bibliography{main}

\end{document}